%% file: joint.tex
\documentclass[10pt,twocolumn,twoside]{IEEEtran}

\usepackage{amsmath, amssymb, amsthm}
\usepackage{mathtools}
\usepackage{hyperref}
\usepackage{bm}
\usepackage{booktabs}
\usepackage{enumitem}
\usepackage{xcolor}
\usepackage{algorithm}
\usepackage{algorithmic}
\usepackage{cite}
\usepackage{graphicx}
\usepackage{multirow}
\usepackage{array}

\input{math_commands}

\providecommand{\vpsi}{\bm{\psi}}
\providecommand{\vbeta}{\bm{\beta}}
\providecommand{\vz}{\bm{z}}
\providecommand{\mTheta}{\bm{\Theta}}

\providecommand{\mK}{\bm{K}}
\providecommand{\mW}{\bm{W}}
\providecommand{\mS}{\bm{S}}
\providecommand{\mU}{\bm{U}}
\providecommand{\gM}{\mathcal{M}}
\providecommand{\gR}{\mathcal{R}}
\providecommand{\gZ}{\mathcal{Z}}
\providecommand{\mPhi}{\bm{\Phi}}

\begin{document}

\title{Differentially Private Synthetic Voltage Phasor Release for Distribution Grids}

\author{%
Andrew Campbell, Chenyue Zhang, Anna Scaglione, Eli Kerr, Merilyn Chesler, and Sean Peisert
\thanks{A.~Campbell, C. Zhang, and A.~Scaglione are with the Department of Electrical and Computer Engineering, Cornell Tech, Cornell University, New York, NY 10044, USA (e-mails: \{ac2458, cz563, as337\}@cornell.edu). Eli Kerr and Merilyn Chesler are with Kevala Inc. S.~Peisert is with Computing Sciences Research, Lawrence Berkeley National Laboratory, Berkeley, CA 94720, USA (e-mail: sppeisert@lbl.gov).}
\thanks{This research was supported in part by the Director, Cybersecurity, Energy Security, and Emergency Response (CESER) office of the U.S.\ Department of Energy, under contract DE-AC02-05CH11231.}
}

\maketitle

\begin{abstract}
Training machine learning models, including Grid Foundation Models (GFMs), requires large volumes of realistic grid data, yet substantial privacy concerns discourage utilities and data providers from sharing load profiles and network parameters. We study the release of synthetic voltage phasor trajectories for distribution grids under differential privacy (DP). We first fit a DP generative model to historical customer loads, then propagate synthetic load trajectories through the AC power flow equations on the true admittance matrix to produce voltage phasors. The central question is whether the randomness already present in the DP synthetic loads is sufficient to protect not only the loads, but also the network topology encoded by the bus admittance matrix. We show that it is. 
The implication is that a corpus of voltage trajectories can be constructed from DP synthetic loads while preserving the statistics of AC power flow, which is critical for training GFMs. This preservation of the power flow statistics stands in contrast to approaches that perturb the admittance matrix directly or inject noise into the voltage outputs, both of which distort the underlying physics.
Concretely, we derive $(\varepsilon,\delta)$-DP guarantees for the released voltage trajectories with respect to the admittance matrix, meaning privacy of the network parameters is obtained without any additional noise mechanism. Our bound depends on the adjacency assumption, the Jacobian of the AC power flow, and the covariance of the synthetic DP-loads. Finally, we present a synthetic voltage generation procedure and an empirical evaluation against Gaussian output-perturbation baselines, demonstrating that our approach provides a clear advantage for enabling GFM training.
\end{abstract}
\section{Introduction}
\label{sec:intro}

The electrification of transportation, the rapid deployment of distributed energy resources (DERs), and the increasing complexity of distribution-grid operations are driving unprecedented demand for data-driven methods in power systems.
Machine learning (ML) models trained on power flow data now support state estimation, fault detection, voltage regulation, load forecasting, and optimal power flow (OPF) approximation~\cite{meiser2024survey,shukla2025ai}.
More ambitiously, the community has begun exploring \emph{Grid Foundation Models} (GFMs)~\cite{Hamann2024FoundationModels}: large-scale pre-trained backbones for diverse grid tasks, which depend on massive heterogeneous datasets spanning load profiles, voltage trajectories, generation patterns, and network topologies across many feeders and operating conditions.
Assembling such datasets, however, confronts a fundamental tension: \emph{the very data that would make GFMs most useful are precisely those that utilities are least willing to share}.

For distribution systems operators (DSOs), two categories of data are confidential:
(1)~load consumption data from advanced metering infrastructure (AMI) sensors, and
(2)~system parameters in the bus admittance matrix~$\mY$.
Voltage phasor measurements depend on both, so their release can let an adversary infer individual consumption behaviors or reconstruct the admittance matrix~\cite{deka2020topology, liao2019urban, deka2023tutorial, yuan2023inverse}.
Bus generation, by contrast, is generally treated as public~\cite{mak2020privacy, sandberg2015differentially, bibi2025comprehensive} given the availability of irradiance measurements and public PV capacity registrations, a convention we adopt throughout. The goal for GFMs is therefore to develop privacy procedures that protect both the load data and the admittance matrix.  

\begin{figure}
    \centering
    \includegraphics[width=0.76\linewidth]{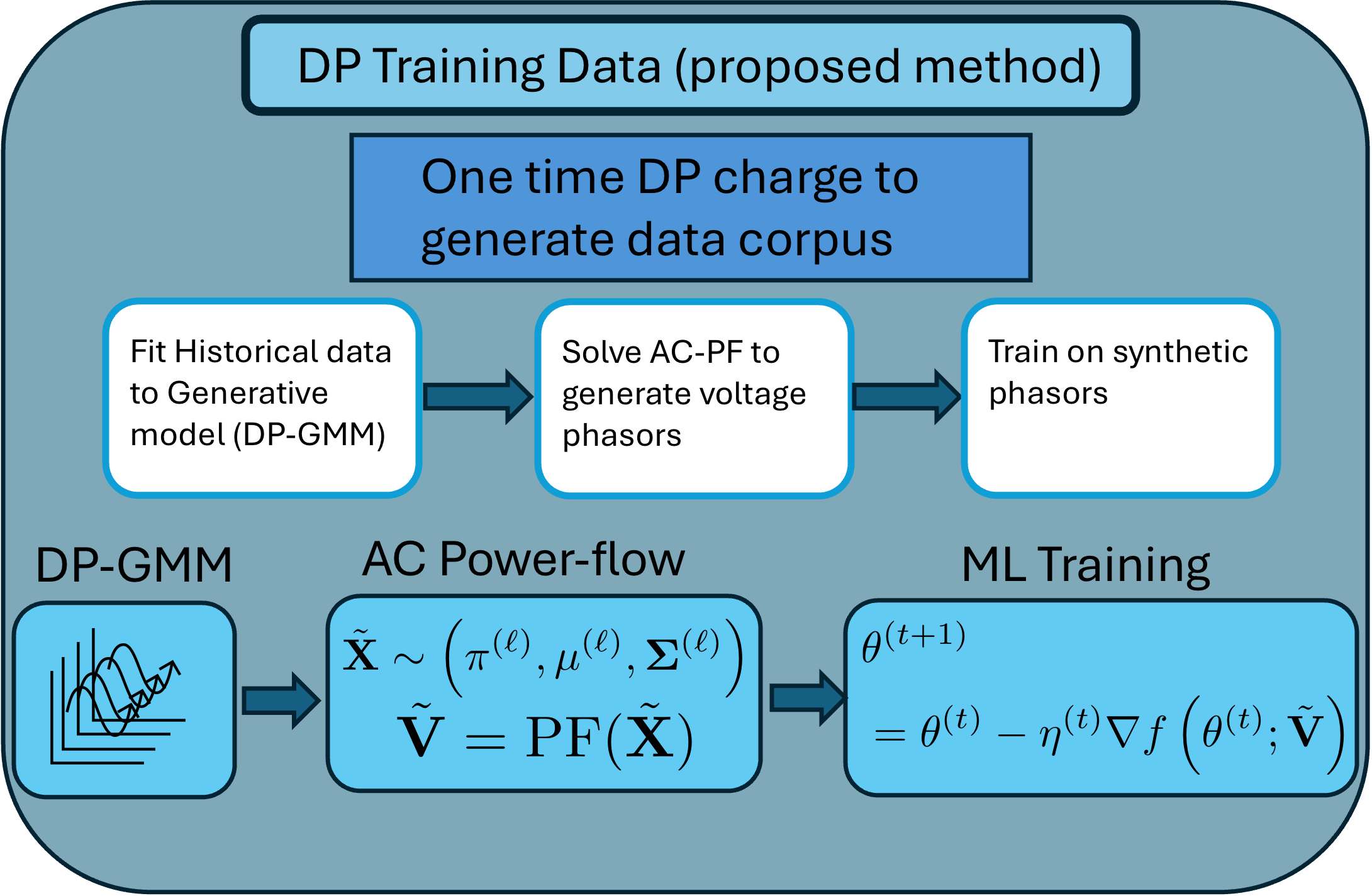}
    \includegraphics[width=.76\linewidth]{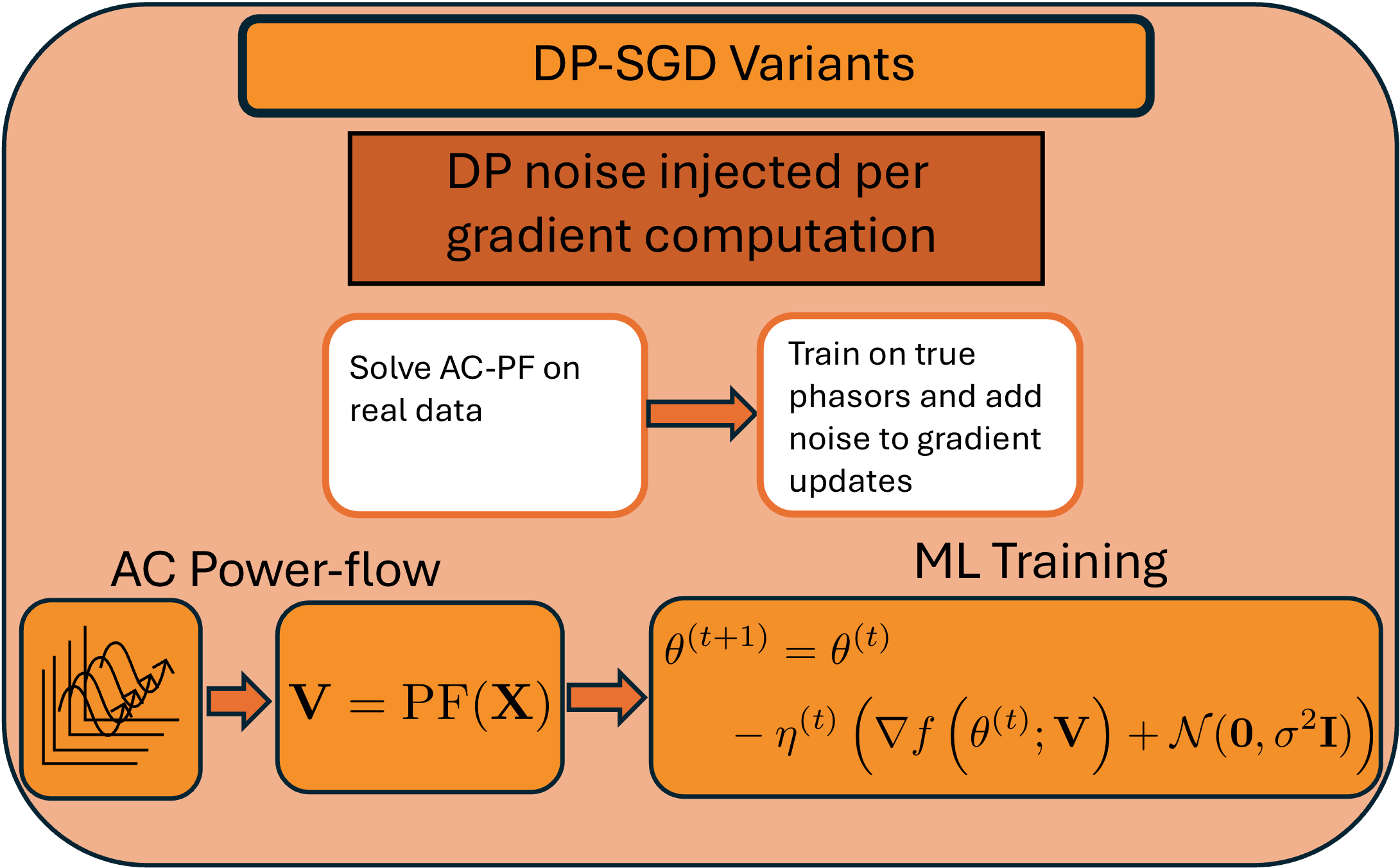}
    \caption{Two DP Paradigms for GFMs. Top is the proposed DP-GMM methodology. Bottom is the standard DP-SDG approach for model training. }
    \label{fig:dpSyn_vs_Dpsgd}
\end{figure}
\begin{table*}[t]
\centering
\caption{Advantages and disadvantages of the two DP paradigms for Grid Foundation Models illustrated in Fig.~\ref{fig:dpSyn_vs_Dpsgd}.}
\label{tab:paradigms}
\scriptsize
\renewcommand{\arraystretch}{1.05}
\setlength{\tabcolsep}{3pt}
\begin{tabular}{@{}>{\raggedright\arraybackslash}p{0.10\linewidth}>{\raggedright\arraybackslash}p{0.46\linewidth}>{\raggedright\arraybackslash}p{0.40\linewidth}@{}}
\toprule
 & \textbf{DP training data (proposed)} & \textbf{DP-SGD variants} \\
\midrule
Advantages
  & (1)~Protects $\mY$ and loads from one mechanism (Thm.~\ref{thm:approx_dp}); (2)~voltages are physics-consistent; (3)~budget spent once at generation, so downstream training is free under post-processing; (4)~DP-GMM preserves data statistics; (5)~task-agnostic corpus.
  & (1)~Simple to implement; (2)~no parametric assumption on loads. \\
\addlinespace[2pt]
Disadvantages
  & (1)~DP-GMM is non-trivial to implement; (2)~utility bounded by DP-GMM fit quality.
  & (1)~Budget composes over iterations; (2)~i.i.d.\ noise destroys correlation; (3)~sample-level guarantee requires lossy group-privacy reduction for $\mY$; (4)~output is a trained model, not a reusable dataset. \\
\bottomrule
\end{tabular}
\renewcommand{\arraystretch}{1.0}
\end{table*}
\subsection{Synthetic Data as the Grid AI Privacy-Preserving Enabler}

Synthetic data generation has emerged as a compelling solution to the data-scarcity and privacy bottleneck~\cite{meiser2024survey,gillioz2025large,turowski2024review}, and recent reviews highlight both the field's rapid expansion and open challenges around standardized evaluation and application-specific realism~\cite{turowski2024review}.
By fitting a generative model to real data and releasing samples from the model rather than the data itself, one can produce arbitrarily large training corpora while limiting disclosure of sensitive information.
Approaches explored so far include Generative Adversarial Networks (GANs) for smart-grid time series~\cite{zhang2018generative}, physics-informed synthetic datasets for transmission grids~\cite{gillioz2025large}, and AI-based techniques for fault classification~\cite{shukla2025ai}.

A critical shortcoming of these approaches is the lack of \emph{formal privacy guarantees}.
Differential privacy (DP)~\cite{dwork2014algorithmic} provides a principled framework for bounding information leakage, but incorporating it into synthetic data pipelines for grid applications remains challenging, and recent surveys emphasize that strong formal privacy with downstream utility is especially difficult in domain-specific settings~\cite{schlegel2025formalprivate}.
The literature is fragmented: some works generate realistic synthetic data with no privacy protection; others add i.i.d.\ noise to real data to achieve DP at the cost of destroying temporal correlations essential for ML training~\cite{gergelyDream, eibl2017differential, barbosa2014lightweight}. Still others train generative models with DP-SGD~\cite{abadi2016deep, huang2022dpwgan}, inheriting the well-known utility degradation of gradient-level perturbation. 
Recent constrained generative approaches improve the physical realism of synthetic power flow data but do not provide topology-level privacy guarantees~\cite{hoseinpour2025constrained}.
See Figure \ref{fig:dpSyn_vs_Dpsgd} for an overview of the DP synthetic data corpus vs DP-SGD for machine learning (ML) tasks. Additionally, in Table \ref{tab:paradigms} we provide the advantages and disadvantages of each approach. 

Prior work has developed along three separate threads: synthetic grid-data generation~\cite{meiser2024survey,turowski2024review}, privacy-preserving smart-grid data release~\cite{bibi2025comprehensive}, and topology/state inference or obfuscation under power flow models~\cite{deka2023tutorial,yuan2023inverse,fioretto2019differential,mak2020privacy}. The closest prior work, and the only one to bridge these threads, is~\cite{dvorkin2023differentially}, which generates DP synthetic power-system datasets via Laplace and exponential mechanisms with convex post-processing to restore downstream-task feasibility.
Their two algorithms protect wind power measurements (for regression) and DC-OPF transmission capacity vectors (for feasibility across OPF models). Our work differs in the protected object and the setting: (i) we protect \emph{loads} and the \emph{full admittance matrix} $\mY$ rather than wind generation or line capacity ratings; (ii) we operate on the \emph{AC} power flow manifold rather than on DC-OPF; and (iii) we retain the true $\mY$ during power flow computation and release voltage phasor trajectories, whereas~\cite{dvorkin2023differentially} releases synthetic network parameters consumed by a separate downstream OPF solve. 
By propagating a DP generative model for loads through the true AC power flow equations, our approach yields physics-consistent voltage trajectories that protect both consumer loads and the admittance matrix with no additional output perturbation.

Table~\ref{tab:comparison} synthesizes the state of the art and highlights the unique position of the present work.
To the best of our knowledge, no prior approach simultaneously achieves all of the following:
\begin{enumerate}[label=(\roman*)]
    \item releases data from a simple, interpretable generative model with strong $(\varepsilon,\delta)$-DP guarantees directly on the model parameters;
    \item minimizes the divergence between the privatized and true data distributions, thereby bounding the degradation of downstream ML models trained on the synthetic data;
    \item enables the release of \emph{physics-consistent system states} (voltage phasors) computed on the \emph{true} network model, without additional output noise; and
    \item provides formal DP guarantees on the system topology ($\mY$) as a byproduct of the same noise mechanism that protects the loads.
\end{enumerate}

Our pipeline builds on two pieces of our prior work: \cite{ravi2022differentially} developed differentially private $K$-Means clustering for AMI data, and~\cite{liu2025differentially} generalized this to a full Gaussian Mixture Model (DP-GMM) that adds calibrated noise to the means, covariances, and mixture weights while minimizing the KL divergence between the privatized and true distributions subject to the $(\varepsilon,\delta)$-DP constraint. 
The GAN-based approach of~\cite{huang2022dpwgan} offers an alternative generative mechanism compatible with our framework, but it trains via DP-SGD~\cite{abadi2016deep}, which provides weaker control over distribution fidelity.
\begin{table*}[t]
\centering
\caption{Comparison of approaches for privacy-preserving synthetic data release in power systems.}
\label{tab:comparison}
\footnotesize
\begin{tabular}{lccccc}
\toprule
\textbf{Approach} & \textbf{Load DP} & \textbf{Topology DP} & \textbf{Distrib.\ fidelity} & \textbf{Unlimited load samples} & \textbf{True $\mY$ in PF} \\
\midrule
i.i.d.\ noise on loads~\cite{gergelyDream,eibl2017differential,barbosa2014lightweight}
  & $\checkmark$ & \texttimes & \texttimes & \texttimes & N/A \\
DP federated learning~\cite{fernandez2022privacy}
  & $\checkmark$ & \texttimes & \texttimes & \texttimes & N/A \\
DP-SGD GAN (DPWGAN)~\cite{huang2022dpwgan}
  & $\checkmark$ & \texttimes & \texttimes$^*$ & $\checkmark$ & N/A \\
Noise on $\mY$ entries~\cite{fioretto2019differential,smith2021realistic}
  & \texttimes & $\checkmark^\dagger$ & \texttimes & \texttimes & \texttimes \\
Noise on OPF outputs~\cite{dvorkin2020differentially}
  & \texttimes & \texttimes & \texttimes & \texttimes & $\checkmark$ \\
Bilevel obfuscation~\cite{mak2019privacy}
  & $\checkmark$ & \texttimes & \texttimes & \texttimes & $\checkmark$ \\
DP $K$-Means clustering~\cite{ravi2022differentially}
  & $\checkmark$ & \texttimes & \texttimes & $\checkmark$ & N/A \\
DP-GMM (KL-optimal)~\cite{liu2025differentially}
  & $\checkmark$ & \texttimes & $\checkmark$ & $\checkmark$ & N/A \\
DP synthetic grid datasets~\cite{dvorkin2023differentially}
  & \texttimes & \texttimes$^\ddagger$ & $\checkmark$ & \texttimes & \texttimes \\
\midrule
\textbf{This work} (DP-Powerflow)
  & $\checkmark$ & $\checkmark$ & $\checkmark$ & $\checkmark$ & $\checkmark$ \\
\bottomrule
\end{tabular}

\vspace{2pt}
{\footnotesize $^*$GAN fidelity is implicit via adversarial training but not measured via a closed-form divergence bound.\\
$^\dagger$Protects $\mY$ under \emph{per-parameter metric DP} rather than the $r$-adjacency used here (Definition~\ref{def:dual_adjacency}); see Remark~\ref{rem:Y_noise_infeasible}.\\
$^\ddagger$Provides DP on a scalar DC-OPF transmission capacity vector $\bar{f}$ under entry-wise $\alpha$-adjacency; the DC power-transfer distribution matrix (encoding graph topology and line susceptances) is held fixed and released in the clear, so no formal DP guarantee applies to topology or impedances.}
\end{table*}

\subsection{From Private Loads to Private Grid States: Two Birds with One Stone}

We address the privacy challenge through a single mechanism that protects loads and voltage phasors simultaneously by: (a)~fitting a DP generative model to the loads (DP-GMM), (b)~sampling synthetic loads from it, and (c)~solving AC power flow on the \emph{true} admittance matrix~$\mY$ with those synthetic loads. The key insight is that the privacy noise in the synthetic loads \emph{propagates through the power flow equations} and induces sufficient randomness in the released voltage phasors to mask~$\mY$, so no additional noise is required on the output voltages.

This ``two birds with one stone'' property has profound implications for GFMs. A utility or data aggregator, can release from a single privacy mechanism both differentially private load profiles and the corresponding voltage phasor trajectories. These voltage phasor trajectories are computed by solving power flow on the true system model with public PV generation profiles. The released voltage data are \emph{physics-consistent} (they satisfy the AC power flow equations on the real network) and \emph{distribution-preserving} (the load distribution is not further corrupted by output noise), making them far more suitable for ML training than data from 
output, or network, perturbation baselines.

This result is a first step toward a broader program. If power flow states can be released privately using synthetic loads on the true grid, then other derived quantities, such as OPF solutions, hosting-capacity assessments, or reliability indices, can also be released using the true system parameters and synthetic inputs, inheriting privacy guarantees from the same DP load model. Formalizing this extension is an important direction for future work.

A naive way to protect both loads and topology would be to first privatize the loads via additive noise~\cite{gergelyDream, eibl2017differential, eibl2018influence, barbosa2014lightweight}, then apply the DP power flow release of~\cite{fioretto2019differential} to protect $\mY$. This benchmark has two shortcomings. First, the load-noising mechanisms of~\cite{gergelyDream, eibl2017differential, eibl2018influence, barbosa2014lightweight} adopt user-level adjacency on individual consumption profiles and say nothing about the admittance matrix, while the adjacency notion of~\cite{fioretto2019differential} is per-parameter metric DP on $\mY$ and does not permit topology changes. In contrast, we consider a Frobenius-ball adjacency on $\mY$ which admits joint perturbations across all entries of $\mY$, capturing genuine topology changes such as line switching rather than just small entry-wise deviations. Second, noising loads and then separately privatizing the induced states compounds distortion, since the two mechanisms inject noise independently at different stages of the power flow. Other DP approaches to grid optimization face related limitations:~\cite{fioretto2019differential} noises line conductances directly;~\cite{mak2019privacy} noises loads and restores feasibility via bilevel optimization; and~\cite{dvorkin2020differentially} noises OPF solutions without releasing loads or states. In contrast, our mechanism obtains formal DP guarantees on the admittance matrix \emph{directly from the noise in the DP synthetic loads}, with no additional perturbation of the power flow outputs. To our knowledge, this work is the first result of its kind.
\subsection{Contributions}
This paper makes the following contributions:
\begin{itemize}
[leftmargin=*,topsep=2pt,itemsep=1pt]
    \item \textbf{Joint DP for loads and topology from a single mechanism.} We show that the randomness in DP synthetic loads, propagated through the power flow equations, suffices to guarantee $(\varepsilon,\delta)$-DP on both the loads and the admittance matrix, with no additional noise on the released voltages.
    \item \textbf{Physically meaningful adjacency for admittance matrices.} We introduce a Frobenius-ball adjacency on $\mY_{\text{full}}$ with a manifold-preservation condition (Definition~\ref{def:dual_adjacency}), which admits joint perturbations across entries---including line switching---rather than fixing the sparsity pattern as in per-parameter metric DP.
    \item \textbf{Closed-form $(\varepsilon,\delta)$-DP bound.} We derive a tractable privacy bound (Theorem~\ref{thm:approx_dp}) that decomposes into an injection-likelihood term and a Jacobian-determinant term, each bounded via the Kron amplification factor, the power flow Jacobian, and the load covariance.
    \item \textbf{Synthetic voltage generation pipeline.} We present Algorithm~\ref{alg:voltage_release} and empirically demonstrate a clear advantage over Gaussian output-perturbation baselines in terms of MSE at matched privacy levels.
\end{itemize}
\section{System Model}
\label{sec:model}

We consider a distribution grid with $N$ buses described by a full complex admittance matrix $\mY_{\mathrm{full}}\in\C^{N\times N}$. We first describe the Kron reduction that eliminates zero-injection buses, since the privacy guarantee is stated in terms of perturbations of $\mY_{\mathrm{full}}$ while the analysis operates on the reduced matrix $\mY$.

\subsection{Kron Reduction and Zero-Injection Buses}
\label{sec:model_kron}
Let $\vv \in \C^N$ denote the vector of complex bus voltage phasors, with $[\vv]_k$ the voltage at bus $k$. Partition the $N$ buses into the set $\gR$ of \emph{retained} buses (those carrying loads or generation, $|\gR|=n$) and the set $\gZ$ of \emph{zero-injection} buses ($|\gZ|=N-n$), with corresponding voltage subvectors $\vv_\gR \in \C^n$ and $\vv_\gZ \in \C^{N-n}$. Block the full admittance matrix accordingly:
\begin{equation}
\label{eq:Yfull_block}
\mY_{\mathrm{full}} = \begin{pmatrix} \mY_{\gR\gR} & \mY_{\gR\gZ} \\ \mY_{\gZ\gR} & \mY_{\gZ\gZ}\end{pmatrix}.
\end{equation}
Since zero-injection buses satisfy, where $s_k$ is the injection at bus $k$, $s_k=0$ for all $k\in\gZ$, the power flow equations at $\gZ$ reduce to the linear system $\mY_{\gZ\gR}\vv_\gR + \mY_{\gZ\gZ}\vv_\gZ = 0$, which, since $\mY_{\gZ\gZ}$ is invertible for any connected network with at least one retained bus, yields
\begin{align}
\label{eq:vZ_recovery}
\vv_\gZ &= -\mY_{\gZ\gZ}^{-1}\mY_{\gZ\gR}\,\vv_\gR \;=:\; \mPhi\,\vv_\gR, \nonumber\\ \mPhi :&= -\mY_{\gZ\gZ}^{-1}\mY_{\gZ\gR}\in\C^{(N-n)\times n}.
\end{align}
Substituting~\eqref{eq:vZ_recovery} into the retained-bus equations gives the exact \emph{Kron-reduced} system
\begin{equation}
\label{eq:Yred_def}
\vs_\gR = \mY\,\vv_\gR, \qquad \mY := \mY_{\gR\gR} - \mY_{\gR\gZ}\mY_{\gZ\gZ}^{-1}\mY_{\gZ\gR}\in\C^{n\times n},
\end{equation}
which is the reduced admittance matrix used throughout.
\begin{remark}
\label{rem:vZ_determined}
Equation~\eqref{eq:vZ_recovery} shows that $\vv_\gZ$ is a deterministic linear function of $\vv_\gR$, so once $\vv_\gR$ is fixed the entire network voltage profile is fixed. That is, the randomness from the DP-GMM propagates to $\vv_\gZ$ through $\mPhi$ and the reduced-system analysis is complete. Moreover, since $|[\vv_\gZ]_k|\le\|[\mPhi]_{k,:}\|_1\cdot\max_j|[\vv_\gR]_j|$ with row sums of $\mPhi$ bounded by construction on radial feeders, $\vv_\gZ$ automatically satisfies ANSI~C84.1 regulation limits whenever $\vv_\gR$ does. No separate treatment of or assumption on $\vv_\gZ$ is needed.
\end{remark}
A perturbation of $\mY_{\mathrm{full}}$ induces a perturbation of the reduced matrix $\mY$, and the induced Frobenius distance on the reduced system can be controlled by a scalar amplification factor $\kappa_{\mathrm{Kron}}$ computable from network data alone.
\begin{definition}[Kron amplification factor]
\label{def:kappa_kron}
For a reference full admittance matrix $\mY_{\mathrm{full}}$ partitioned as in~\eqref{eq:Yfull_block}, define
\begin{equation}
\label{eq:kappa_kron_def}
\kappa_{\mathrm{Kron}} := 1 + 2\,\|\mY_{\gR\gZ}\|_{\mathrm{op}}\|\mY_{\gZ\gZ}^{-1}\|_{\mathrm{op}} + \|\mY_{\gR\gZ}\|_{\mathrm{op}}^2\,\|\mY_{\gZ\gZ}^{-1}\|_{\mathrm{op}}^2.
\end{equation}
\end{definition}
\begin{corollary}
\label{cor:kron_perturbation}
For any two full admittance matrices $\mY_{\mathrm{full}},\mY_{\mathrm{full}}'$ with $\|\mY_{\mathrm{full}}-\mY_{\mathrm{full}}'\|_F < r$, the Kron-reduced matrices $\mY,\mY'$ given by~\eqref{eq:Yred_def} satisfy, to leading order in $\|\mY_{\mathrm{full}}-\mY_{\mathrm{full}}'\|_F$,
\begin{equation}
\label{eq:kappa_bound}
\|\mY - \mY'\|_F \;\le\; \kappa_{\mathrm{Kron}}\,\|\mY_{\mathrm{full}} - \mY_{\mathrm{full}}'\|_F \;<\; \kappa_{\mathrm{Kron}}\,r.
\end{equation}
\end{corollary}

Throughout the remainder of the paper, the symbol $\mY$ (without subscript) always denotes the Kron-reduced admittance matrix~\eqref{eq:Yred_def}, while $\mY_{\mathrm{full}}$ denotes the original unreduced matrix. The voltage vector $\vv_t\in\C^n$ refers to the retained-bus voltages $\vv_{\gR,t}$. Next we introduce the AC power flow model on the reduced system.
\subsection{AC Power Flow}
The nonlinear AC power flow equations relate the complex injection vector $\vs_t \in \C^n$ to voltages and admittances through the map $F_{\mY}:\C^n\rightarrow\C^n$ where:
\begin{equation}
    \label{eq:powerflow}
    \vs_t :=F_{\mY}(\vv_t)= \diag(\vv_t)\left(\bar{\mY}\,\bar{\vv}_t+\bar{\vb}\right),
\end{equation}
with $\vv_t \in \C^n$ denoting the retained-bus voltage trajectory at time $t$ and $\vb \in \C^n$ the constant-current offset induced by the slack bus. Specifically, letting $v_{\mathrm{slack}}\in\C$ denote the (fixed) slack-bus voltage and indexing the retained buses to exclude the slack,
\begin{equation}
\label{eq:b_def}
\vb := [\mY_{\mathrm{full}}]_{\mathrm{retained},\,\mathrm{slack}}\,v_{\mathrm{slack}},
\end{equation}
so $\vb$ is known exactly from the network model and the slack-bus reference.   
Since the privacy analysis requires real-valued Jacobians and determinants, we work with the real representation $\tilde{F}_{\mY} : \R^{2n} \to \R^{2n}$, defined by
\begin{equation}
\label{eq:real_rep}
\tilde{F}_{\mY} : \begin{pmatrix} \Real[\vv_t] \\ \Imag[\vv_t] \end{pmatrix} \mapsto \begin{pmatrix} \Real[\vs_t] \\ \Imag[\vs_t] \end{pmatrix}, \quad \vs_t = F_{\mY}(\vv_t).
\end{equation}
We denote its Jacobian by $\mJ_{\tilde{F}_{\mY}}(\vv) \in \R^{2n \times 2n}$.
Since $\tilde{F}_\mY$ is quadratic in $\vv_t$ and generally non-invertible, we restrict the analysis to a ``good'' voltage set defined in Section~\ref{sec:model_good_volt} on which $\tilde{F}_\mY$ is invertible and $\mJ_{\tilde{F}_{\mY}}$ is well defined.

\subsection{Injection Model}
\label{sec:model_injection}
We adopt the standard net injection convention where generation is positive and load is negative, so the total complex injection decomposes as
\begin{equation}
\label{eq:inj_decomp}
\vs_t = \vs_t^g - \vs_t^p,
\end{equation}
where $\vs_t^p \in \C^n$ is the complex load vector with $\Real\{[\vs_t^p]_k\} \ge 0$ the active load magnitude and $\Imag\{[\vs_t^p]_k\} \ge 0$ the reactive load magnitude. Since active load is well modeled by log-normal distributions~\cite{liu2025differentially,munkhammar2014characterizing,mey2021prediction,kuusela2015practical}, we partition buses into $L \ll n$ load classes $\mathcal{C}_1, \ldots, \mathcal{C}_L$, each representing a consumer type (industrial, residential, etc.). Let $\vp_t\in\R^n$ denote the active load at time $t$. For each bus $k\in\mathcal C_\ell$, the $T$-dimensional log-load vector over the release horizon is jointly Gaussian:
\begin{align}
\label{eq:loadclasses}
\Big(\log [\vp^{\ell}_{1}]_k,\,\ldots,\,\log [\vp^{\ell}_T]_{k}\Big)^\top
\sim \mathcal N\!\big(\bm{\mu}^{(\ell)},\,\bm{\Sigma}^{(\ell)}_{T}\big),\nonumber \\
\forall\, k\in\mathcal C_\ell,\ \ell\in[L],
\end{align}
where $\bm{\mu}^{(\ell)}\in\mathbb{R}^T$ is the time-varying class mean profile and $\bm{\Sigma}^{(\ell)}_{T}\in\mathbb{R}^{T\times T}$ with $\bm{\Sigma}^{(\ell)}_{T}\succ 0$ is the temporal covariance. The log-load vectors $\{(\log [\vp^{\ell}_1]_{k},\ldots,\log [\vp^{\ell}_T]_{k})\}_{k\in\mathcal C_\ell}$ are independent across buses, with common distribution~\eqref{eq:loadclasses}. Each load bus $k \in \gC_{\ell}$ has a fixed power factor angle $\theta^{(\ell)}_{k}\in[-\pi/2,\pi/2]$, so reactive power satisfies
\begin{align}
    [\Imag\{\vs_t^p\}]_k = \tan(\theta_k^{(\ell)})\,[\vp_t^\ell]_k. \label{eq:powerfactor}
\end{align}

Let $\gG\subset\left\{1,\hdots,n\right\}$ denote the buses with Photovoltaic (PV) generation. For each $k\in\gG$, $\abs{[\vs_t^g]_k} = \gamma_k h_t^g$, where $h_t^g \geq 0$ is the public irradiance and $\gamma_k > 0$ is installed capacity. Smart inverters couple the reactive injection to local voltage through a volt-var curve $\phi_k([\vv_t]_k) : \R_+ \rightarrow [-\pi/2, \pi/2]$, giving
\begin{equation}
    \label{eq:voltvar}
    [\vs_t^g]_k = \gamma_k h_t^g \cdot e^{j\phi_k(\abs{[\vv_t]_k})}.
\end{equation}
Letting $\Gamma = \diag(\gamma_k)_{k \in \mathcal{G}}\in\R^{n\times n}$ and $[\vh_t^g(\vv_t)]_k = h_t^g e^{j\phi_k(\abs{[\vv_t]_k})}\in\C$ for $k \in \mathcal{G}$, the combined system becomes the fixed-point problem:\vspace{-.2cm}
\begin{equation}
\label{eq:fixedpoint}
\diag(\vv_t)\,\left(\bar{\mY}\,\bar{\vv}_t+\bar{\vb}\right)
= \overbrace{\Gamma \vh_t^g(\vv_t)}^{\vs_t^g(\vv_t)} \;-\; \vs_t^p.
\end{equation}
Thus~\eqref{eq:fixedpoint} defines the load-to-voltage map only implicitly, and it does not necessarily admit a unique smooth solution everywhere. The next subsection restricts to operating conditions under which this map is well-defined and locally invertible.

\subsection{``Good" Voltage Set}
\label{sec:model_good_volt}
To address the invertibility concerns raised by~\eqref{eq:fixedpoint}, we restrict voltages to a ``good'' set defined as follows.
\begin{definition}[``Good" Voltage Set]
\label{def:S0}
Fix $V_{\min}, V_{\max} > 0$. Define
\begin{align}
\mathcal S_0 := \bigg\{\vv\in\C^n:\ V_{\min}\leq|[\vv]_k|\leq V_{\max},\;\forall\,k\bigg\}.\label{eq:S_0}
\end{align}
\end{definition}
The geometric constraints in $\mathcal S_0$ reflect standard voltage regulation limits (e.g., ANSI~C84.1 specifies $|v_k|\in[0.95,1.05]$~p.u.). However, restricting voltages to $\mathcal S_0$ alone does not guarantee that the power flow map is locally invertible or that the implied loads are physically meaningful. We therefore require the admittance matrix to satisfy the following.
\begin{assumption}[Normal Operating Conditions]
\label{def:feasible_class}
The feeder operates under normal conditions; specifically:

\smallskip
\noindent\textbf{(i) Voltage stability margin.}
The system is not operating at or beyond its voltage stability limit, i.e.\ the power flow Jacobian is nonsingular on $\mathcal{S}_0$:
\[
  \inf_{v\in\mathcal{S}_0,\; Y\in\mathcal{Y}_{\mathrm{feas}}}
  \sigma_{\min}\!\bigl(\mJ_{\tilde F_Y}(v)\bigr) > 0.
\]
This is automatic for any normally loaded feeder under ANSI~C84.1 regulation limits, and by the implicit function theorem is equivalent to local uniqueness of the high-voltage solution branch (i.e. no voltage collapse).

\smallskip
\noindent\textbf{(ii) Rated loading envelope.}
Every retained-bus voltage profile $\vv_\gR\in\mathcal{S}_0$ corresponds to active loads within the design range $[p^{(\ell)}_{\min},\,p^{(\ell)}_{\max}]$ for each load class $\ell$:
\[
  p^{(\ell)}_{\min}
  \;\leq\; \bigl[\hat \vp_\mY(\vv_\gR)\bigr]_k
  \;\leq\; p^{(\ell)}_{\max},
  \qquad \forall\, \vv_\gR\in\mathcal{S}_0,\; k\in\mathcal{C}_\ell,
\]
where $[\hat \vp_\mY(\vv)]_k := \Re([\vs^g(\vv)]_k) - \Re\bigl([\vv]_k (\sum_{j=1}^n \overline{\mY_{kj}}\,\overline{[\vv]_j}+\overline{\vb}_k)\bigr)$ is the implied active load at bus $k$. This excludes both unloaded and overloaded feeders, with $p^{(\ell)}_{\min},p^{(\ell)}_{\max}$ read from the feeder's load study or AMI history. By Remark~\ref{rem:vZ_determined} the zero-injection voltages are controlled automatically.

\smallskip
\noindent\textbf{(iii) IEEE~1547-2018 volt-var compliance (when applicable).}
When smart inverters operate under a volt-var curve $\phi_k$, the combined Jacobian remains nonsingular on $\mathcal{S}_0$: $\inf_{v\in\mathcal{S}_0}\sigma_{\min}(J^{\mathrm{eff}}_{\tilde F_Y}(v)) > 0$. IEEE~1547-2018 requires $|\phi'_k(u)| < 1/(\gamma_k h^g_{\max})$, which automatically enforces this for any certified deployment; absent volt-var coupling, condition~(iii) reduces to~(i).
\end{assumption}

The margins $[p^{(\ell)}_{\min},p^{(\ell)}_{\max}]$ in condition~(ii) serve two purposes: they guarantee that the log-normal likelihood ratio in the privacy analysis is bounded, and they define the support of the truncated log-normal distribution from which synthetic loads are sampled (see Definition~\ref{def:mechanism} and Algorithm~\ref{alg:voltage_release}).

\subsection{Pushed-Forward Voltage Distribution}
\label{sec:model_push}
The net injection $\vs_t$ is random only through the $n$-dim active-load vector $\vp_t$ since the reactive component $\vq_t = \tan(\mTheta)\vp_t$, is a deterministic function of $\vp_t$ via the fixed power factors~\eqref{eq:powerfactor}, and $\vs_t^g(\vv_t)$ is a deterministic function of $\vv_t$ via~\eqref{eq:voltvar}. The distribution of $\vs_t$ on $\R^{2n}$ is therefore supported on an $n$-dim submanifold, and the voltage density must be written as a surface density on the corresponding $n$-dim image in $\R^{2n}$.

Under Assumption~\ref{def:feasible_class}, \eqref{eq:fixedpoint} admits a unique smooth solution $\vv = G_\mY(\vp)$ for each admissible $\vp\in\R^n$, yielding a smooth injective map
\begin{equation*}
G_\mY : \R^n \to \R^{2n}, \quad \gM_\mY := G_\mY(\R^n) \subset \R^{2n},
\end{equation*}
whose image $\gM_\mY$ is the $n$-dim submanifold of $\mY$-admissible voltages. Taking real parts of~\eqref{eq:fixedpoint} and using $\Real[(\mI + j\tan\mTheta)\vp] = \vp$ gives the explicit closed form of the inverse:
\begin{equation}
\label{eq:Ginv_closed_form}
    [G_\mY^{-1}(\vv)]_k = \Real\{[\vs^g(\vv)]_k\} - \Real\bigl\{[\vv]_k[\bar\mY\bar\vv+\bar\vb]_k\bigr\}, \quad k\in[n].
\end{equation}
By the implicit function theorem applied to~\eqref{eq:fixedpoint},
\begin{equation}
\label{eq:DG_ift}
    DG_\mY(\vp) = -\bigl(\mJ^{\mathrm{eff}}_{\tilde F_\mY}(\vv)\bigr)^{-1}\mR, \quad \mR:= \begin{pmatrix}\mI \\ \tan\mTheta\end{pmatrix}\in\R^{2n\times n},
\end{equation}
evaluated at $\vv = G_\mY(\vp)$, where $\mR$ depends only on the power factors.

Let $f_\vp$ denote the joint density of the active log-load vector $\vp$ under the log-GMM model~\eqref{eq:loadclasses}. Then the voltage $\vv = G_\mY(\vp)$, as a random element of $\R^{2n}$ supported on $\gM_\mY$, has density, given by the $n$-dim change-of-variables formula:
\begin{equation}
    \label{eq:pushforward}
    p_{\vv}(\vv \mid \mY) = f_\vp\bigl(G_\mY^{-1}(\vv)\bigr)\,|J_\mY(\vv)|, \quad \vv\in\gM_\mY,
\end{equation}
where the scalar volume factor $|J_\mY(\vv)|$ is the reciprocal square root of the Gramian determinant of $DG_\mY$:
\begin{equation}
\label{eq:Jn_def}
|J_\mY(\vv)| := \sqrt{\det\bigl(DG_\mY^\top DG_\mY\bigr)}\bigg|_{\vp = G_\mY^{-1}(\vv)}.
\end{equation}
Crucially, \eqref{eq:pushforward} rewrites the distribution of the voltage phasors in terms of the distribution of the active loads, giving a procedure for bounding the privacy loss of the voltage phasors w.r.t.\ the admittance matrix using the noise in the DP loads. When volt-var coupling is present, the generation injection depends on the voltage through the volt-var curve and $\mJ_{\tilde{F}_{\mY}}$ must be replaced by an effective Jacobian $\mJ^{\mathrm{eff}}_{\tilde{F}_{\mY}}$ throughout. Since this substitution does not impact the logic of the analysis, we proceed w.l.o.g.\ with $\mJ_{\tilde{F}_{\mY}}$ and defer the details to Appendix~\ref{app:Mtilde_closed_form_proof}.

\section{Privacy Preliminaries}
\label{sec:privacy}
We introduce the threat model, DP definitions and adjacency, and then the privacy mechanism.
\subsection{Threat Model}
\label{sec:threat_model}
We adopt the standard DP threat model used in prior work on privacy-preserving grid data release~\cite{fioretto2019differential,smith2021realistic,dvorkin2020differentially,mak2020privacy,mak2019privacy,sandberg2015differentially}: the adversary is computationally unbounded, may possess arbitrary auxiliary information about loads, generation, and network topology, and observes every published output. Privacy is quantified by the distinguishability of the system state under adjacent admittance matrices $\mY,\mY'$. This contrasts with DP graph release~\cite{d2026sok,dvorkin2023differentially}, whose published artifact is the privatized graph itself. Our mechanism instead releases the voltage phasor trajectories $\{\tilde{\vv}_\tau\}_{\tau=1}^T$ computed on the true~$\mY$, so the push-forward distribution~\eqref{eq:pushforward} is the central object of the analysis.


\subsection{Differential Privacy and Adjacency}

We adopt \emph{Probabilistic DP} (PDP)~\cite{machanavajjhala2008privacy}, which implies the standard $(\varepsilon,\delta)$-DP of~\cite{dwork2014algorithmic} and connects directly with hypothesis testing over probability distributions:
\begin{definition}[$(\varepsilon,\delta)$-PDP~\cite{machanavajjhala2008privacy}]
    \label{def:pdp}
    Let $\mD,\mD'\in\mathcal{D}$ be adjacent data sets in database $\mathcal{D}$, and let $\mathcal{M}$ be a query mechanism over $\mathcal{D}$ with output $\tilde{q}\sim f(\tilde{q}|\mD)$. Then $\mathcal{M}$ is $(\varepsilon,\delta)$-PDP if
    \begin{align}
        \Pr \left(\left|\ln\dfrac{f(\tilde{q}|\mD)}{f(\tilde{q}|\mD')}\right|>\varepsilon\right)\leq \delta.
    \end{align}
\end{definition}
In our DP power flow setting the database is $\gY_{\mathrm{feas}}$ and we consider $\mY_{\mathrm{full}},\mY_{\mathrm{full}}'$ adjacent if they satisfy:
\begin{definition}[Physical $r$-adjacency]
\label{def:dual_adjacency}
Two full admittance matrices $\mY_{\mathrm{full}},\mY_{\mathrm{full}}'$ with Kron reductions $\mY,\mY'\in\gY_{\mathrm{feas}}$ are \emph{$r$-adjacent} if
\begin{align}
\|\mY_{\mathrm{full}} - \mY_{\mathrm{full}}'\|_F &< r, \label{eq:adj_frob} \\
\gM_\mY &= \gM_{\mY'}, \label{eq:adj_manifold}
\end{align}
where $\gM_\mY := G_\mY(\R^n)$ is the set of retained-bus voltages in $\gS_0$ reachable by the mechanism under $\mY$ with the fixed power factors $\tan\mTheta$. 
For a reference $\mY_{\mathrm{full}}$,
\begin{align}
\mathcal N_{r}(\mY_{\mathrm{full}}) := \bigl\{\mY_{\mathrm{full}}':\ \|\mY_{\mathrm{full}}-\mY_{\mathrm{full}}'\|_F < r,\ \gM_{\mY'} = \gM_{\mY}\bigr\}.\nonumber
\end{align}
\end{definition}

Condition~\eqref{eq:adj_frob} is a \emph{physical} adjacency notion: the mechanism protects against distinguishing any two feeders whose line admittances (including switch states) differ by a total Frobenius distance less than $r$. By Corollary~\ref{cor:kron_perturbation}, any such pair induces reduced matrices satisfying $\|\mY-\mY'\|_F<\kappa_{\mathrm{Kron}}r$, which is the quantity that actually enters the privacy analysis of Section~\ref{sec:priv_analysis}. The manifold condition~\eqref{eq:adj_manifold} is measure-theoretic: it ensures that the pushforward measures induced by the mechanism share a common support on the reduced system, so that their Radon--Nikodym derivative exists, which is a requirement for the LLR to be well defined. 
Linearizing the power-factor consistency residual around $\mY$ yields $n$ homogeneous linear constraints on the induced $\Delta\mY$, so $\mathcal{N}_r(\mY_{\mathrm{full}})$ corresponds (via Corollary~\ref{cor:kron_perturbation}) to an $(n^2-n)$-dim subspace of the $\Delta\mY$-space intersected with a Frobenius ball of radius $\kappa_{\mathrm{Kron}} r$. This defines a rich family of admittance matrices against which the mechanism is formally DP.

In contrast, prior DP power flow work~\cite{fioretto2019differential,smith2021realistic} adopts per-parameter metric DP, which protects each admittance entry independently up to a scalar threshold and thereby fixes $\mY$'s sparsity pattern a priori. On the other hand,~\cite{dvorkin2023differentially} protects a scalar DC-OPF transmission capacity vector $\bar{\mathbf{f}}$ and thus defines their adjacency as entry-wise $\alpha$-adjacency on $\bar{\mathbf{f}}$, and neither the admittance matrix $\mY$ nor the loads enter the privacy guarantee. Our Frobenius-ball adjacency is instead applied to $\mY_{\mathrm{full}}$ itself and is coupled with DP on the load distribution via the pushforward through the true $\mY$, yielding joint protection of both objects. This dual protection is not offered by any prior DP power flow formulation. 

\subsection{The Mechanism and Its Privacy Guarantee}

The data owner (utility) holds the true admittance matrix $\mY$ and uses a standard power flow solver (e.g., OpenDSS) to compute voltage phasors from synthetic DP load inputs. We assume the DP loads are log-normal mixtures as in~\eqref{eq:loadclasses}, which we call the DP-GMM; the specific DP-GMM procedure is irrelevant to the analysis and deferred to Section~\ref{sec:results}. 

\begin{definition}[DP-Powerflow]
\label{def:mechanism}
Given the true admittance matrix $\mY$, public irradiance data $\{h_\tau^g\}$, and a DP-GMM fitted to historical loads, the mechanism $\mathcal{M}_\mY$ operates as follows:
\begin{enumerate}[label=(\alph*)]
    \item Draw a synthetic load vector $\tilde{\vs}_t^p$ from the DP-GMM, truncated to $[p_{\min}^{(\ell)},p_{\max}^{(\ell)}]$ per class.
    \item Solve the power flow equations~\eqref{eq:fixedpoint} with the true $\mY$ to obtain $\tilde{\vv}_t$;
    \item Release $\tilde{\vv}_t$.
\end{enumerate}
\end{definition}
See Algorithm~\ref{alg:voltage_release} for details. Because the synthetic load $\tilde{\vs}_t^p$ is random and never disclosed, the released voltage $\tilde{\vv}_t$ is a random variable whose distribution depends on $\mY$ through the power flow map, and no additional voltage noise is needed since the randomness of the synthetic load input is the privacy mechanism. Formally, $\mathcal{M}_\mY$ at time $t$ outputs a sample from the pushforward density $p_\vv(\cdot \mid \mY)$ in~\eqref{eq:pushforward}, and the DP guarantee with respect to $\mY$ follows from bounding the log-likelihood ratio across adjacent admittance matrices.

\begin{algorithm}[t]
\caption{Synthetic Voltage Phasor Release}
\label{alg:voltage_release}
\begin{algorithmic}[1]
\REQUIRE Historical log-load data $\{\log \vp_k^{(i)}\}_{k\in[n],\,i\in[d]}$ with $\log \vp_k^{(i)}\in\R^T$; true admittance matrix $\mY$; public irradiance $\{h_\tau^g\}_{\tau=1}^T$; power factor angles $\{\theta_k^{(\ell)}\}$; number of classes $L$; per-class load margins $p_{\min}^{(\ell)},p_{\max}^{(\ell)}$; privacy budget $\varepsilon_{\mathrm{load}}$.

\medskip
\STATE \textbf{Phase 1: Offline model fitting}
\STATE $\mathcal{C}_1,\ldots,\mathcal{C}_L
       \;\leftarrow\;
       \textsc{Partition}\!\bigl(\{\log \vp_k^{(i)}\}_{k,i}\bigr)$
\FOR{$\ell = 1,\ldots,L$}
    \STATE Class data:
           $\mathcal{D}_\ell
           \leftarrow
           \bigl\{\log \vp_k^{(i)} : k\in\mathcal{C}_\ell,\;i\in[d]\bigr\}
           \subset\R^T$
    \STATE Fit DP Gaussian to\\ $\mathcal{D}_\ell$:
    $\bigl(\tilde{\vmu}^{(\ell)},\,
            \tilde{\mSigma}_T^{(\ell)}\bigr)
    \;\leftarrow\;
    \textsc{DP-GMM}(\mathcal{D}_\ell,\,\varepsilon_{\mathrm{load}})$
\ENDFOR

\medskip
\STATE \textbf{Phase 2: Synthetic voltage generation}
\FOR{each bus $k\in\mathcal{C}_\ell$, $\ell\in[L]$}
    \STATE Draw $T$-dimensional log-load from the truncated distribution:\\
           $\tilde{\vxi}_k
           \sim\mathcal{N}\!\bigl(\tilde{\vmu}^{(\ell)},\,
                                   \tilde{\mSigma}_T^{(\ell)}\bigr)\;\Big|\;
                                   \exp(\tilde{\vxi}_k)\in
                                   [p_{\min}^{(\ell)},\,p_{\max}^{(\ell)}]^T$.
    \STATE Set active load:
           $[\tilde\vp_k]_\tau
           \leftarrow
           \exp([\tilde{\vxi}_k]_\tau)$
           for $\tau\in[T]$.
    \STATE Reactive load:
           $[\tilde\vq_k]_\tau
           \leftarrow
           \tan(\theta_k^{(\ell)})\,[\tilde\vp_k]_\tau$
           for $\tau\in[T]$.
\ENDFOR
\FOR{$\tau = 1,\ldots,T$}
    \STATE Form load injection:
           $\tilde{\vs}_\tau^p
           \leftarrow \tilde{\vp}_\tau + j\,\tilde{\vq}_\tau$.
    \STATE Solve power flow:\\
           $\tilde{\vv}_\tau
           \leftarrow
           \textsc{PowerFlowSolve}(\mY,\,
           \tilde{\vs}_\tau^p,\,\vs_{\tau}^g)$
\ENDFOR
\STATE \textbf{Release}
       $\{\tilde{\vv}_\tau\}_{\tau=1}^T$.
\end{algorithmic}
\end{algorithm}

\section{Privacy Analysis}
\label{sec:priv_analysis}
We develop the $(\varepsilon,\delta)$-PDP guarantee by decomposing the log-likelihood ratio (LLR) into an injection-likelihood term (Term~I) and a Jacobian-determinant term (Term~II), bounding each, and combining them via a $\chi^2$ concentration inequality on the whitened drawn loads.

\subsection{Log-Likelihood-Ratio Decomposition}
\label{sec:llr}
Let $\Lambda(\vv_{1:T}; \mY, \mY') := \log\bigl(p_\vv(\vv_{1:T}\mid\mY)/p_\vv(\vv_{1:T}\mid\mY')\bigr)$ denote the LLR between the pushforward densities under $\mY$ and $\mY'$. By Definition~\ref{def:pdp}, the mechanism is $(\varepsilon,\delta)$-PDP if $\Pr[|\Lambda|>\varepsilon]\leq\delta$. For $\mY_{\mathrm{full}}'\in\gN_r(\mY_{\mathrm{full}})$, the common-support condition~\eqref{eq:adj_manifold} ensures that $G_{\mY'}^{-1}$ is well-defined on $\gM_\mY$ and the two surface densities share support. Using~\eqref{eq:pushforward}, $\Lambda$ decomposes as
\begin{align}
    \label{eq:llr_decomp}
    \Lambda(\vv_{1:T}; \mY, \mY') &= \underbrace{\log \frac{f_\vp\bigl(\vp_{1:T}^\mY\bigr)}{f_\vp\bigl(\vp_{1:T}^{\mY'}\bigr)}}_{\text{Term I: injection likelihood}}+ \underbrace{\sum_{t=1}^T\log\frac{|J_\mY(\vv_t)|}{|J_{\mY'}(\vv_t)|}}_{\text{Term II: Jacobian ratio}},
\end{align}
where $\vp_{1:T}^{\mY} := (G_\mY^{-1}(\vv_1),\ldots,G_\mY^{-1}(\vv_T))$ and $\vp_{1:T}^{\mY'} := (G_{\mY'}^{-1}(\vv_1),\ldots,G_{\mY'}^{-1}(\vv_T))$ are the implied active-load trajectories.
Term~I measures how much the implied active-load trajectory shifts when the admittance matrix changes from $\mY$ to $\mY'$, evaluated under the class-$\ell$ log-normal load density with log-covariance $\mSigma^{(\ell)}$.

Term~II captures how the change alters the $n$-dim surface volume element via the Gramian of $DG_\mY$. We bound each term separately and combine.
\subsection{Term I: Injection Sensitivity}
\label{sec:termI}

At a fixed voltage trajectory $(\vv_1,\ldots,\vv_T)\in\gS_0^T$, admittances $\mY$ and $\mY' = \mY - \Delta\mY$ imply different active loads at each bus, with $\|\Delta\mY\|_F<\kappa_{\mathrm{Kron}} r$ by Corollary~\ref{cor:kron_perturbation}. Generation terms cancel, leaving $[\hat{p}_{\mY}(\vv_t)]_k - [\hat{p}_{\mY'}(\vv_t)]_k = -\Re\{[\vv_t]_k[\overline{\Delta\mY}\bar{\vv}_t]_k\}$. The following proposition bounds Term~I.

\begin{proposition}
\label{prop:termI}
Let $\Delta\mY := \mY - \mY'$ denote the induced perturbation of the Kron-reduced matrices, which by Corollary~\ref{cor:kron_perturbation} satisfies $\|\Delta\mY\|_F \le \kappa_{\mathrm{Kron}} r$ for any $r$-adjacent $\mY_{\mathrm{full}},\mY_{\mathrm{full}}'$. For each class $\ell\in[L]$, define the entrywise sensitivity constant
\begin{align}
\label{eq:dl_def}
d_\ell := \frac{V_{\max}^2\sqrt{d_{\max}}}{p_{\min}^{(\ell)}},
\end{align}
where $d_{\max} := \max_k |\{j : \mY_{kj}\neq 0\}|$ is the maximum node degree, and the \emph{precision sum}
\begin{align}
\label{eq:gamma_def}
\gamma^{(\ell)} := \vone^\top|(\mSigma^{(\ell)})^{-1}|\vone = \sum_{t,t'} |[(\mSigma^{(\ell)})^{-1}]_{tt'}|.
\end{align}
Define the uniform whitened-shift bound
\begin{align}
\label{eq:barpsi_def}
\bar\psi^{\,2} &:= (\kappa_{\mathrm{Kron}}\,r)^{2}\sum_{\ell=1}^L d_\ell^{\,2}\,\gamma^{(\ell)},
\end{align}
and the $\chi^2_{nT}$ tail factor
\begin{align}
\label{eq:tau_def}
\tau(\delta_R) &:= \sqrt{\,nT + 2\sqrt{nT\,\log(1/\delta_R)} + 2\log(1/\delta_R)\,}.
\end{align}
Then for any $\delta_R\in(0,1)$, with probability at least $1-\delta_R$ over the DP-GMM, Term~I satisfies
\begin{align}
\label{eq:termI_bound}
&|\mathrm{Term~I}_{1:T}|\nonumber\\
&~~~~\le \bar\psi\,\tau(\delta_R) + \tfrac{1}{2}\bar\psi^{\,2}
+ \kappa_{\mathrm{Kron}}\,r\sum_{\ell=1}^L d_\ell\sqrt{\gamma^{(\ell)}|\gC_\ell|}\,\sqrt{\vone^\top\mSigma^{(\ell)}\vone}.
\end{align}
\end{proposition}

For a proof see Appendix~\ref{app:termI_proof}.

With the injection likelihood (Term~I) now bounded, it remains to bound the Jacobian determinant ratio (Term~II).

\subsection{Bounding Term~II: Jacobian Determinant Ratio}
\label{sec:termII}

Term~II measures how changing the admittance matrix alters the volume element of the power flow map. The analysis relies on the normalized Jacobian $\tilde{\mM}$, obtained by factoring out a diagonal voltage matrix from the Wirtinger Jacobian of the power flow map (see Appendix~\ref{app:termII_proof} for the derivation). Specifically, the Wirtinger Jacobian factors as $\mJ = \mD(\vv)\,\tilde{\mM}$, where
\begin{align}
\label{eq:Dv_def}
\mD(\vv) := \diag(v_1,\ldots,v_n,\bar{v}_1,\ldots,\bar{v}_n) \in \C^{2n\times 2n},
\end{align}
and the normalized Jacobian takes the form
\begin{align}
\label{eq:Mtilde_def}
\tilde{\mM} = \begin{pmatrix} \diag(\vs/\vv^2) & \bar{\mY} \\ \mY & \diag(\bar{\vs}/\bar{\vv}^2) \end{pmatrix},
\end{align}
where $s_i = v_i(\bar{\mY}\bar{\vv}+\bar{\vb})_i$ is the complex power injection at bus $i$.

Define the worst-case operator norm $\|\tilde{\mM}^{-1}\|_\star := \sup_{\vv\in\gS_0,\,\mY'\in\mathcal{N}_r(\mY_{\mathrm{full}})} \|\tilde{\mM}(\vv,\mY')^{-1}\|_{\mathrm{op}}$ and the geometric constant $C_\star := \sqrt{2}(1 + \sqrt{n}\,V_{\max}/V_{\min})$. Although a supremum, $\|\tilde{\mM}^{-1}\|_\star$ admits a closed-form upper bound from network data alone (see Appendix~\ref{app:Mtilde_closed_form_proof} for the full derivation). However, since this bound is naturally conservative, our evaluation uses the Monte Carlo calibration of $\norm{\tilde{\mM}^{-1}}_*$ which is detailed in Remark~\ref{rem:mc_calibration}.

\begin{proposition}
\label{prop:termII}
Define the admissibility parameter
\begin{align}
\label{eq:alpha_def}
\alpha := \|\tilde{\mM}^{-1}\|_\star\,C_\star\,\kappa_{\mathrm{Kron}}\,r,
\end{align}
with $\|\tilde{\mM}^{-1}\|_\star$ and $C_\star$ as defined above. Under the admissibility condition $\alpha < 1/4$,
\begin{align}
\label{eq:termII_trajectory}
|\mathrm{Term~II}_{1:T}| \le \bar{\Lambda}_{\mathrm{II}} := \frac{T\sqrt{n}\,\alpha(2+\alpha)}{2(1-4\alpha)}.
\end{align}
\end{proposition}

For a proof see Appendix~\ref{app:termII_proof}. The bound is controlled by the conditioning of $\tilde{\mM}$, the adjacency radius (through $\kappa_{\mathrm{Kron}} r$), the voltage bounds (through $C_\star$), and the trajectory length $T$. For small $\alpha$ it simplifies to $\bar\Lambda_{\mathrm{II}}\approx T\sqrt{n}\,\alpha$. Combining both terms gives the main privacy guarantee which is provided below.

\begin{theorem}
\label{thm:approx_dp}
Under the system model of Section~\ref{sec:model}, the log-load model~\eqref{eq:loadclasses}, the feasibility conditions of Assumption~\ref{def:feasible_class}, and the admissibility condition $\alpha < 1/4$ with $\alpha$ as in~\eqref{eq:alpha_def}, the mechanism is $(\varepsilon,\delta)$-PDP with respect to physical $r$-adjacency, with
\begin{align}
\label{eq:eps_delta}
\varepsilon = B + \bar\psi\,\tau(\delta),
\end{align}
where the total deterministic bias is
\begin{align}
\label{eq:B_total}
B &= \frac{T\sqrt{n}\,\alpha(2+\alpha)}{2(1-4\alpha)} + \frac{1}{2}\bar\psi^{\,2} \nonumber\\
&\quad + \kappa_{\mathrm{Kron}}\,r\sum_{\ell=1}^L d_\ell\sqrt{\gamma^{(\ell)}|\gC_\ell|}\;\sqrt{\vone^\top\mSigma^{(\ell)}\vone},
\end{align}
the uniform whitened-shift bound is
\begin{align}
\label{eq:barpsi_thm}
\bar\psi^{\,2} = (\kappa_{\mathrm{Kron}}\,r)^{2}\sum_{\ell=1}^L d_\ell^{\,2}\,\gamma^{(\ell)},
\end{align}
and the $\chi^2_{nT}$ tail factor is
\begin{align}
\label{eq:tau_thm}
\tau(\delta) = \sqrt{\,nT + 2\sqrt{nT\,\log(1/\delta)} + 2\log(1/\delta)\,},
\end{align}
with $d_\ell = V_{\max}^2\sqrt{d_{\max}}/p_{\min}^{(\ell)}$ and $\gamma^{(\ell)} = \vone^\top|(\mSigma^{(\ell)})^{-1}|\vone$.
\end{theorem}

For a proof see Appendix~\ref{app:thm_approx_proof}. Theorem~\ref{thm:approx_dp} delivers a formal $(\varepsilon,\delta)$-DP guarantee on the admittance matrix $\mY$ from the DP-synthetic loads alone, yielding our contribution that no additional perturbation is necessary. The bound has three pieces. The Jacobian term $T\sqrt{n}\,\alpha(2+\alpha)/(2(1-4\alpha))$ captures volume-element distortion of the power flow map and is controlled by the adjacency radius $r$ through $\alpha$. The deterministic bias scales as $(\kappa_{\mathrm{Kron}}r)^2$ and is likewise driven by $r$. The tail term $\bar\psi\,\tau(\delta)$ captures the stochastic fluctuation of Term~I, where $\bar\psi$ depends on the DP-GMM covariance $\mSigma^{(\ell)}$ and $\tau(\delta)\approx\sqrt{nT}$ for $nT\gg\log(1/\delta)$. 

The practitioner has two levers. The adjacency radius $r$ is a sensitivity parameter that sets how broad a class of topology perturbations the mechanism protects against, and it controls all three terms of the bound. 
The DP-GMM covariance $\mSigma^{(\ell)}$ enters only the tail and bias where a larger covariance decreases the precision sum $\gamma^{(\ell)}$ and thus a noisier DP-GMM yields stronger topology-level privacy. Because the Jacobian term does not depend on $\mSigma^{(\ell)}$, it imposes a floor on $\varepsilon$ that no amount of DP-GMM noise can cross. Driving $\varepsilon$ below this floor requires tightening $r$.

The remaining quantities are fixed by the feeder but shape how forgiving the bound is. A well-conditioned admittance matrix has small $\|\tilde{\mM}^{-1}\|_\star$, which tightens $\alpha$ and shrinks the Jacobian term. Poorly connected feeders with near-singular $\tilde{\mM}$ inflate $\alpha$ and can push the admissibility condition $\alpha<1/4$ toward its boundary. A wide voltage window $[V_{\min}, V_{\max}]$ inflates $C_\star = \sqrt{2}(1+\sqrt{n}V_{\max}/V_{\min})$, which also inflates $\alpha$, so feeders operating close to nominal voltage give tighter guarantees than those with loose regulation. 
The load margins $[p_{\min}^{(\ell)}, p_{\max}^{(\ell)}]$ appear through $d_\ell = V_{\max}^2\sqrt{d_{\max}}/p_{\min}^{(\ell)}$ where load classes with small $p_{\min}^{(\ell)}$, which correspond to feeders with lightly loaded buses, amplify the bias and tail terms.  Consequently, the practitioner should set $p_{\min}^{(\ell)}$ and $p_{\max}^{(\ell)}$ as tight as feasibility allows, since wider margins trade directly against the privacy bound through $d_\ell$. Finally, when volt-var coupling is present, Theorem~\ref{thm:approx_dp} holds verbatim upon replacing $\tilde{\mM}$ by the effective Jacobian $\tilde{\mM}_{\mathrm{eff}} = \tilde{\mM} - \mD(\vv)^{-1}\mJ_{\Gamma\vh}$ in the definition of $\alpha$.

\begin{remark}[Monte Carlo calibration of $\bar{\Lambda}_{\mathrm{II}}$]
\label{rem:mc_calibration}
The worst-case $\|\tilde{\mM}^{-1}\|_\star$ used in $\bar{\Lambda}_{\mathrm{II}}$ is taken over all of $\gS_0$, which may be overly conservative. In practice, one can replace $\|\tilde{\mM}^{-1}\|_\star$ by a probabilistic threshold $\mu_0$ and absorb the exceedance probability into~$\delta$. Specifically, define
\begin{align}
\delta_M := \Pr_{\vp\sim f_\vp}\!\bigl(\exists\,t:\|\tilde{\mM}(\vv_t,\mY)^{-1}\|_{\mathrm{op}} > \mu_0'\bigr), \nonumber\\ \mu_0' = \frac{\mu_0}{1+\mu_0\,C_\star\,\kappa_{\mathrm{Kron}}\,r},
\end{align}
where the shift $\mu_0\to\mu_0'$ accounts for the worst-case perturbation over $\mathcal{N}_r(\mY_{\mathrm{full}})$ via the induced reduced-system radius $\kappa_{\mathrm{Kron}} r$. The probability $\delta_M$ can be estimated by drawing $N_{\mathrm{cal}}$ trajectories from the DP-GMM, solving power flow with the utility's own $\mY$, and counting exceedances (with Clopper--Pearson confidence intervals). The mechanism is then $(\varepsilon,\delta+\delta_M)$-PDP with $\bar{\Lambda}_{\mathrm{II}}$ computed using $\mu_0$ in place of $\|\tilde{\mM}^{-1}\|_\star$, where $\varepsilon$ is computed from Theorem~\ref{thm:approx_dp}.
\end{remark}

\section{Empirical Evaluation}
\label{sec:results}

We conduct two complementary evaluations. The first measures distributional fidelity of the released voltage phasors through the Wasserstein-1 distance between the privatized and true voltage magnitude distributions, swept over a grid of target $\varepsilon$ values. The second evaluates downstream machine-learning utility by training a three-layer multilayer perceptron (MLP) on privatized voltage data for a missing-data recovery task and reporting test-set MSE. 
The Wasserstein distance upper-bounds the generalization gap of any downstream model trained on the synthetic data~\cite{lopez2018generalization}. Making this distance small is therefore particularly valuable for training foundation models, since it directly limits how much worse a downstream model can perform when trained on synthetic rather than true data.
The MLP experiment in turn demonstrates that naive DP-SGD training~\cite{abadi2016deep} is insufficient for the extensive training requirements of grid foundation models and indicates the benefit of our proposed method over naive noise approaches.

Both experiments utilize the IEEE~123-bus test feeder with a full year of PV and loads provided by the OEDI dataset\cite{OEDI_Dataset_5773}. The distribution feeder is compiled in OpenDSS, and we apply the Kron reduction of Section~\ref{sec:model_kron} to eliminate the zero-injection buses, yielding a reduced admittance matrix $\mY\in\C^{n\times n}$ and the constant-current offset $\vb = [\mY_{\mathrm{full}}]_{\mathrm{retained},\,\mathrm{slack}}\,v_{\mathrm{slack}}$ as in~\eqref{eq:b_def}. Historical load and generation profiles are obtained from the feeder's loadshape definitions at 15-minute resolution ($T=96$ samples per day).

\begin{table*}[t]
\centering
\caption{Per-sample $\ell_2$ sensitivities and noise standard deviations for each mechanism.}
\label{tab:sensitivities}
\footnotesize
\renewcommand{\arraystretch}{1.4}
\begin{tabular}{llccc}
\toprule
\textbf{Mechanism} & \textbf{Noise target} & \textbf{Per-sample $\Delta_2$} & \textbf{Noise $\sigma$} & \textbf{Protects} \\
\midrule
DP-GMM $\to$ PF (proposed) & --- & --- (Thm.~\ref{thm:approx_dp}) & none & $\mY$ + loads \\
Gauss.\ on voltages & voltages & $\Delta_2^{(\mY)}$ & $\frac{\sqrt{T}\,\Delta_2^{(\mY)}\sqrt{2\ln(1.25/\delta)}}{\varepsilon}$ & $\mY$ \\[4pt]
Gauss.\ on $\mY$ & $\mY$ entries & $r$ & $\frac{r\sqrt{2\ln(1.25/\delta)}}{\varepsilon}$ & $\mY$ \\[4pt]

Joint voltage noise & voltages & $\max\!\bigl(\Delta_2^{(\mY)},\,\Delta_2^{(\mathrm{load}\to v)}\bigr)$ & $\frac{\sqrt{T}\,\Delta_2^{(\mathrm{joint})}\sqrt{2\ln(1.25/\delta)}}{\varepsilon}$ & $\mY$ + loads \\[4pt]
DP-GMM + Gauss.\ volt.\ & voltages & $\Delta_2^{(\mY)}$ & $\frac{\sqrt{T}\,\Delta_2^{(\mY)}\sqrt{2\ln(1.25/\delta_{\mY})}}{\varepsilon_{\mY}}$ & $\mY$ + loads \\[4pt]
Noisy loads + Gauss.\ volt.\ & load entries & $\sqrt{n_L T}\,\Delta_{\mathrm{load}}$ & $\frac{\sqrt{n_L T}\,\Delta_{\mathrm{load}}\sqrt{2\ln(1.25/\delta_{\mathrm{load}})}}{\varepsilon_{\mathrm{load}}}$ & \multirow{2}{*}{$\mY$ + loads} \\
 & voltages & $\Delta_2^{(\mY)}$ & $\frac{\sqrt{T}\,\Delta_2^{(\mY)}\sqrt{2\ln(1.25/\delta_{\mY})}}{\varepsilon_{\mY}}$ & \\
\bottomrule
\end{tabular}
\renewcommand{\arraystretch}{1.0}
\vspace{4pt}

{\footnotesize See~\eqref{eq:sens_voltage_Y}--\eqref{eq:sens_voltage_load} for $\Delta_2^{(\mY)}$ and $\Delta_2^{(\mathrm{load}\to v)}$. Here $r$ is the physical Frobenius radius on $\mY_{\mathrm{full}}$ (Definition~\ref{def:dual_adjacency}) and $\kappa_{\mathrm{Kron}}$ is the Kron amplification factor~\eqref{eq:kappa_kron_def}.}
\end{table*}
\subsection{Mechanisms Under Comparison}
\label{sec:mechanisms}

We compare several voltage-release mechanisms under the same cumulative privacy budget. For target $(\varepsilon,\delta)$, each mechanism ensures the release of $Td$ voltage phasor vectors ($d$ days) is $(\varepsilon,\delta)$-DP with respect to the quantity it protects. We write $\Delta_{\mathrm{load}} := p_{\max} - p_{\min}$ for the global active-load range and $n_L$ for the number of load buses.
Per-sample sensitivities and noise calibrations follow standard Gaussian mechanism arguments (Appendix~\ref{app:sensitivity_proofs}) and are collected in Table~\ref{tab:sensitivities}. Two sensitivities recur:
\begin{align}
\label{eq:sens_voltage_Y}
\Delta_2^{(\mY)} &= \frac{V_{\max}^2\sqrt{n}\,\kappa_{\mathrm{Kron}}\,r\,\|\tilde{\mM}^{-1}\|_\star}{V_{\min}}, \\
\label{eq:sens_voltage_load}
\Delta_2^{(\mathrm{load}\to v)} &= \frac{\sqrt{2}\,\Delta_{\mathrm{load}}\,\|\tilde{\mM}^{-1}\|_\star}{V_{\min}}.
\end{align}
The mechanisms are:
\begin{itemize}[leftmargin=*,topsep=2pt,itemsep=1pt]
\item \emph{DP-GMM $\to$ PF (proposed)}: Algorithm~\ref{alg:voltage_release} with synthetic loads at budget $\varepsilon_{\mathrm{load}}$ passed through AC power flow on the true~$\mY$. $\varepsilon$ is computed via Theorem~\ref{thm:approx_dp} with $\|\tilde{\mM}^{-1}\|_\star$ calibrated via Remark~\ref{rem:mc_calibration}. We select $\varepsilon_{\mathrm{load}}$ that minimizes Wasserstein distance subject to the target $\varepsilon$.
\item \emph{Noise-free baseline}: AC power flow with true loads and~$\mY$. 
\item \emph{Joint Voltage Noise}: a single Gaussian mechanism on the voltage output protects both $\mY$ and loads, with budget $\varepsilon=\min(\varepsilon_{\mathrm{load}},\varepsilon)$, where sensitivity is given by the max sensitivity of the load and voltage release respectively.
\item \emph{DP-GMM + Gaussian Voltage Noise}: DP-GMM loads at budget $\varepsilon_{\mathrm{load}}$ passed through the true~$\mY$, with an additional Gaussian mechanism on the voltages at budget $\varepsilon$ calibrated to the sensitivity of $\mY$.
\item \emph{Noisy Loads + Gaussian Voltage Noise}: replaces the DP-GMM with i.i.d.\ Gaussian noise on the load matrix at $\varepsilon_{\mathrm{load}}$, clipped and passed through the true~$\mY$, with Gaussian noise added to the voltages at budget $\varepsilon$.
\item \emph{DP-SGD}~\cite{abadi2016deep}: included only in the ML experiment since it protects the trained model rather than the released data. Matching the $\mY$-level budget $(\varepsilon,\delta)$ requires a group-privacy reduction $\varepsilon_{\mathrm{sample}}=\varepsilon/N$ (with adjusted $\delta_{\mathrm{sample}}$)~\cite{dwork2014algorithmic} to account for one feeder's $\mY$ affecting all $N$ samples we then convert to zero-concentrated DP~\cite{bun2016concentrated} via $\rho=(\sqrt{\ln(1/\delta_{\mathrm{sample}})+\varepsilon_{\mathrm{sample}}}-\sqrt{\ln(1/\delta_{\mathrm{sample}})})^2$ and calibrate the per-step Gaussian noise multiplier under subsampling rate $q=B/N$ ($B=64$, $N=1000$).
\end{itemize}
\begin{remark}
\label{rem:Y_noise_infeasible}
Prior work on admittance-matrix privacy~\cite{fioretto2019differential,smith2021realistic} adopts \emph{per-parameter metric DP} with constraint-based feasibility-restoring post-processing, which fixes the sparsity pattern of $\mY$ a priori and is therefore not comparable to our $r$-adjacency. Similarly,~\cite{dvorkin2023differentially} provides DP with respect to thermal constraints and does not protect the topology. The natural remaining comparison is a Gaussian mechanism applied directly to $\mY$, but this is not viable without a feasibility-restoring step. Even at $r=10^{-3}$ and $\varepsilon=100$, OpenDSS failed to converge because unstructured perturbation destroys the sparsity, symmetry, and shunt structure required for Assumption~\ref{def:feasible_class}. We therefore omit this comparison and leave a feasibility-restoring variant to future work.
\end{remark}

\begin{figure}
    \centering
\includegraphics[width=0.85\linewidth]{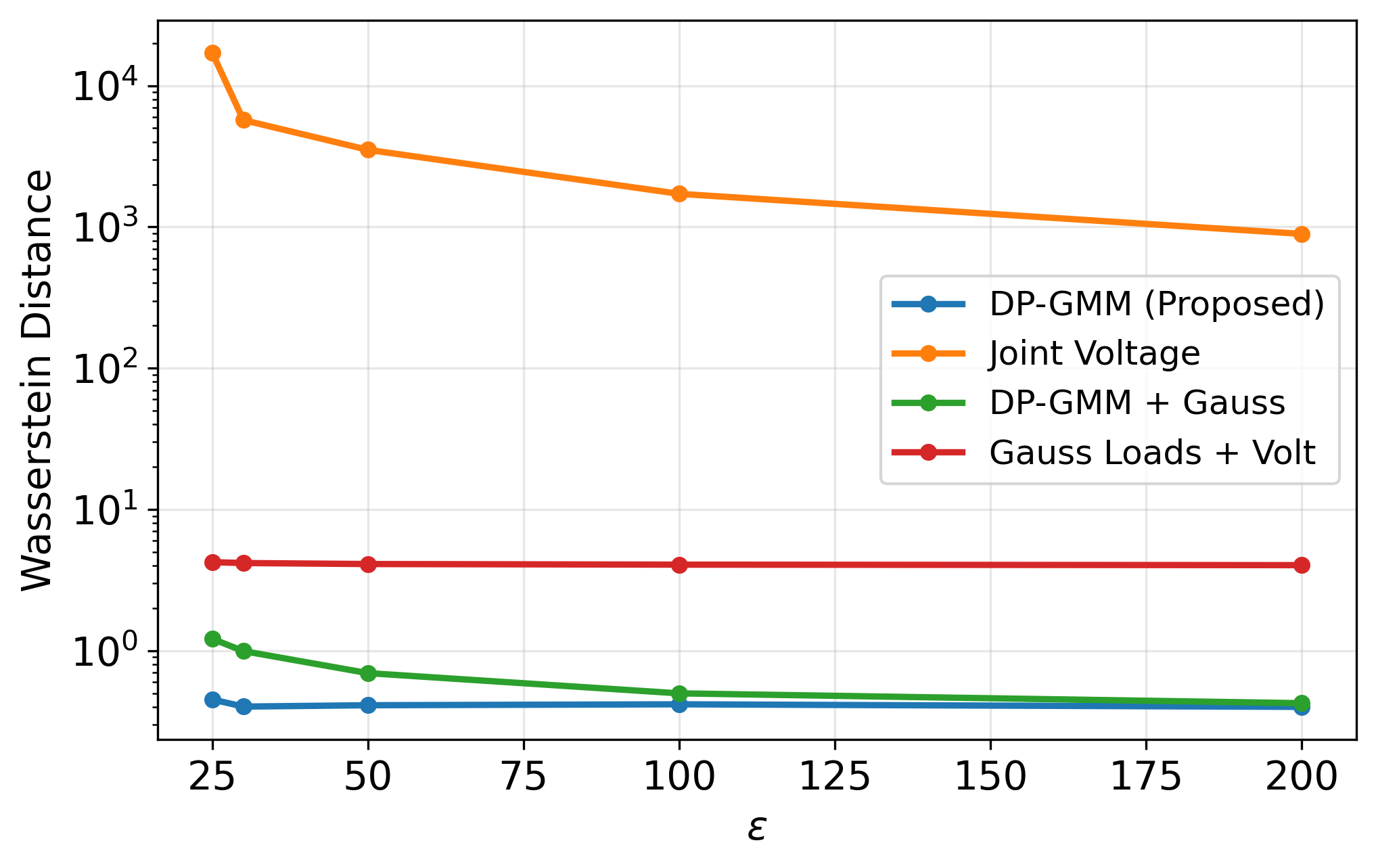}
    \caption{Wasserstein-1 metric as a function of $\varepsilon$.}
    \label{fig:wass}
\end{figure}

\begin{figure*}
    \centering
    \includegraphics[width=0.45\linewidth]{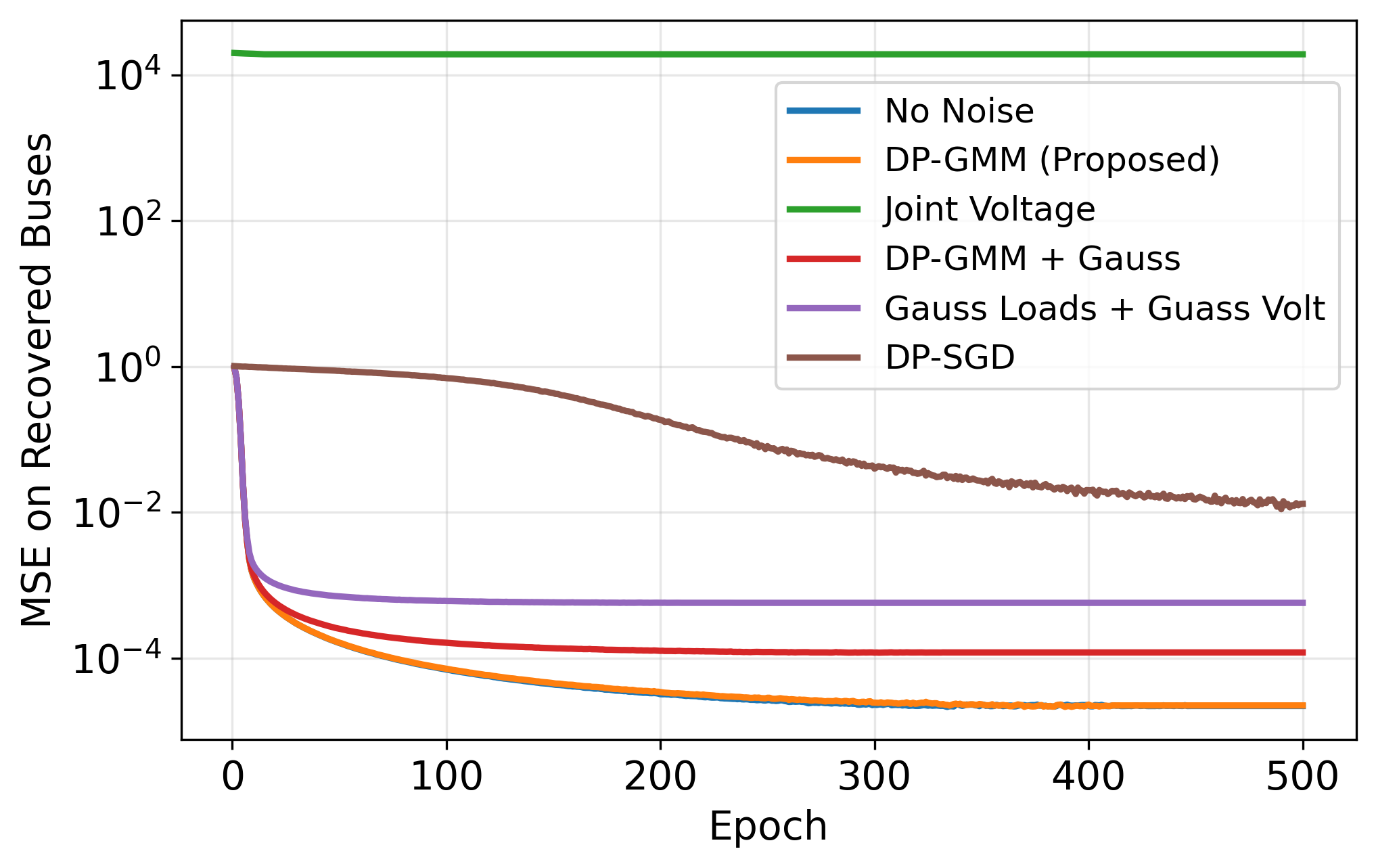}
    \includegraphics[width=.45\linewidth]{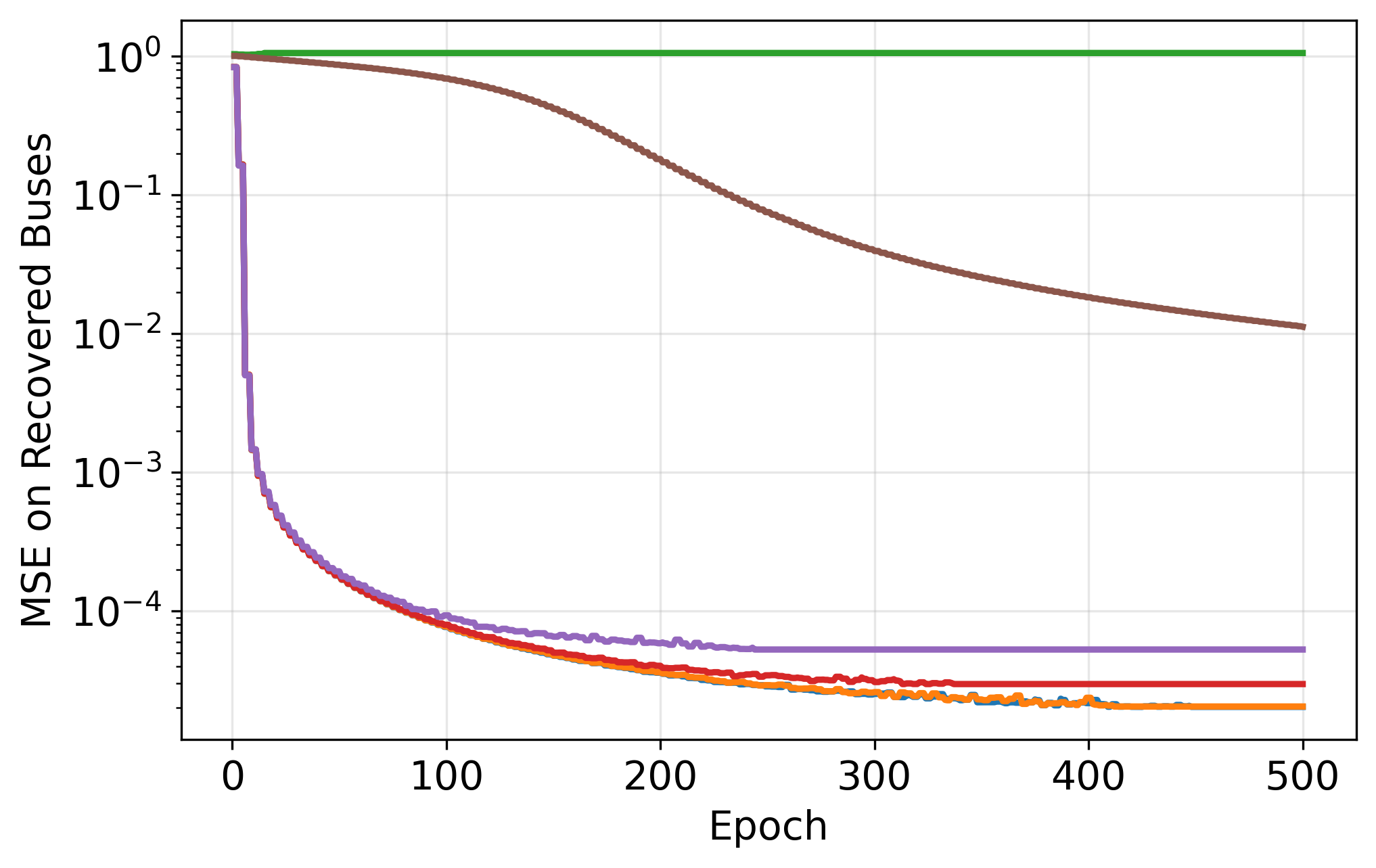}
    \caption{Training (left) and test (right) MSE on the masked entries of the per-bus missing-data recovery task (Section~\ref{sec:ml_experiment}), plotted against training epoch. All methods operate at the same $\mY$-level privacy budget $(\varepsilon,\delta)$ and share the same true-data evaluation and test sets. Curves show the mean across runs.}
    \label{fig:ml_exp}
\end{figure*}
\subsection{Wasserstein Distance Experiment}
\label{sec:wasserstein}
We sweep over a grid of per-sample$(\vv_t\in\R^n)$ $\varepsilon$ values $\{25,\,30,\,50,\,100,\,200\}$ and, for each value, generate voltage trajectories from all mechanisms over a fixed number of days, using the $\delta$ produced by the DP-GMM Monte Carlo calibration (Remark~\ref{rem:mc_calibration}) as the common $\delta$. Voltages are generated by solving AC power flow via OpenDSS with the true network model. For the proposed DP-GMM mechanism, we search over $\varepsilon_{\mathrm{load}}$ values, evaluate Theorem~\ref{thm:approx_dp} via Monte Carlo, and select the configuration that minimizes the Wasserstein distance subject to the target $\varepsilon$ for~$\mY$. Each configuration is evaluated over 20 independent Monte Carlo runs, and reported values are averaged across runs.

The Wasserstein-1 distance is computed between the flattened daily voltage magnitude vectors of each mechanism and the true (non-private) voltages. Figure~\ref{fig:wass} shows that the DP-GMM approach achieves a substantially smaller Wasserstein distance than the Joint Voltage and Noisy Loads + Gaussian Voltage baselines, and a moderate improvement over the DP-GMM + Gaussian output-perturbation variant that narrows at large $\varepsilon$. This improvement over the latter reflects the effective ``free'' privacy revealed by our analysis. The proposed mechanism attains the same fidelity without additional Gaussian corruption of the voltages.

\subsection{Machine Learning Experiment}
\label{sec:ml_experiment}
We evaluate downstream ML utility via a per-bus missing-data recovery task. Each sample is a single bus's voltage magnitude trajectory over a contiguous window of $d_w$ time steps, i.e., a vector in $\R^{d_w}$. We set $d_w = 48$ (12 hours at 15-minute resolution) with a 3 hour overlap with subsequent samples. From each trajectory, a random contiguous block of length $L = \lfloor\sqrt{0.25}\,d_w\rfloor = 24$ is masked, and the task is to reconstruct the masked block from the observed context on the same bus. A three-layer MLP (with ReLU activations and hidden dimension~$32$) is trained to predict the full trajectory from the partial observation, with MSE loss evaluated on the masked entries only. The evaluation and test sets are drawn from the \emph{true} (non-private) voltage distribution and are shared across all methods (including DP-SGD), which are compared at the same $\mY$-level privacy budget $(\varepsilon,\delta)$ as in Section~\ref{sec:mechanisms}.

For each day of $T=96$ voltage vectors generated by a single feeder, the release is $(\varepsilon_{\mathrm{day}},\delta)$-DP at the $\mY$ level, so a training corpus spanning $d$ days has total cost $\varepsilon = d\,\varepsilon_{\mathrm{day}}$ by basic composition. Training uses Adam (standard) or per-sample-clipped SGD (DP-SGD) with early stopping on the shared evaluation MSE. Each method is trained over 20 independent runs with different random seeds, and Fig.~\ref{fig:ml_exp} reports the mean training and test MSE across runs. Fig.~\ref{fig:ml_exp} shows two takeaways. First, DP-SGD converges slowly because its per-iteration budget composition forces substantial noise at matched $\mY$-level privacy. Second, the proposed mechanism tracks the error-free baseline closely, reflecting the value of preserving load correlation structure via DP-GMM and using the true admittance matrix in power flow. Finally, Fig.~\ref{fig:ml_exp} shows a clear gain from the proposed method over the DP-GMM + Gaussian output-perturbation variant. The separation between the two curves is entirely attributable to the ``free'' privacy enabled by our analysis, mirroring the gap observed in the Wasserstein experiment. 
\section{Conclusion}
\label{sec:conclusion}
The practical barrier to training GFMs is the availability of data. Privacy concerns over customer loads and network topology keep utility datasets siloed. We showed that a single privacy mechanism, namely a DP-GMM fitted to historical loads and propagated through the true AC power flow, yields voltage phasor trajectories that are simultaneously $(\varepsilon,\delta)$-DP with respect to both the loads and the admittance matrix $\mY$, with the guarantee on $\mY$ obtained \emph{for free} from the DP noise already required by the load model. Empirical results on the IEEE 123-bus feeder confirm that the proposed mechanism closely tracks the no-noise baseline, while Gaussian output-perturbation baselines and DP-SGD degrade substantially at matched privacy budgets. The principle generalizes beyond voltage release. Any quantity computed from loads and a fixed network, such as OPF solutions, hosting-capacity assessments, or reliability indices, inherits privacy guarantees from the same DP load model, which we view as the natural next step toward privacy-preserving Grid Foundation Models.

\bibliographystyle{ieeetr}
\bibliography{refs}

\begin{appendix}

\subsection{Supporting Lemma}
\label{app:supporting}

\begin{lemma}[Log is Lipschitz on a positive interval]
\label{lem:log_lipschitz}
If $x,x'\in[p_{\min},p_{\max}]$ with $p_{\min}>0$, then $|\log x-\log x'|\le |x-x'|/p_{\min}$. \emph{Proof.} By the mean value theorem, $\log x-\log x'=(x-x')/\xi$ for some $\xi\in[p_{\min},p_{\max}]$, so $|1/\xi|\le 1/p_{\min}$.
\end{lemma}

\subsection{Proof of Proposition~\ref{prop:termI}}
\label{app:termI_proof}

The proof proceeds in four steps: (i)~derive the exact per-bus LLR via the midpoint formula, (ii)~pass to whitened coordinates, (iii)~bound $\sum_k\|\vpsi_k(\vv)\|^2$ uniformly in $\vv$, and (iv)~apply Cauchy--Schwarz together with a $\chi^2$ tail on $\|\vz^{\mY}\|_2$.

\emph{Step 1: Midpoint formula.}
For bus $k\in\gC_\ell$, let $\vxi_k(\vv) := \log\vp_k^{\mY}(\vv)$ and $\vxi'_k(\vv) := \log\vp_k^{\mY'}(\vv)$ denote the \emph{implied} log-load vectors at voltage $\vv$ under the two admittance matrices. Define $\vDelta_k(\vv) := \vxi_k(\vv) - \vxi'_k(\vv)$ and the midpoint $\bar{\vxi}_k(\vv) := \tfrac{1}{2}(\vxi_k(\vv)+\vxi'_k(\vv))$.

The log-density of the multivariate log-normal at $\vp_k^{\mY}(\vv)$ is
\[
\log f(\vp_k^{\mY}(\vv)) = \mathrm{const} - \tfrac{1}{2}(\vxi_k - \vmu^{(\ell)})^\top(\mSigma^{(\ell)})^{-1}(\vxi_k - \vmu^{(\ell)}) - \vone^\top\vxi_k.
\]
In the ratio, the normalisation constant cancels. The quadratic term yields $-\vDelta_k^\top(\mSigma^{(\ell)})^{-1}(\bar{\vxi}_k - \vmu^{(\ell)})$, using $\vxi'_k + \tfrac{1}{2}\vDelta_k = \bar{\vxi}_k$. The Jacobian term $-\vone^\top\vxi_k$ contributes $-\vone^\top\vDelta_k$. Combining:
\begin{equation}
\label{eq:midpoint_formula_app}
\Lambda_I^{(k)} = \log\frac{f(\vp_k^{\mY}(\vv))}{f(\vp_k^{\mY'}(\vv))} = \vDelta_k(\vv)^\top(\mSigma^{(\ell)})^{-1}( \vmu^{(\ell)} + \mSigma^{(\ell)}\vone-\bar{\vxi}_k(\vv)).
\end{equation}

\emph{Step 2: Whitened coordinates.} Let $\vbeta^{(\ell)} := (\mSigma^{(\ell)})^{1/2}\vone$ and $\vz_k^{\mY} := (\mSigma^{(\ell)})^{-1/2}(\vxi_k(\vv) - \vmu^{(\ell)})$ and $\vpsi_k(\vv) := (\mSigma^{(\ell)})^{-1/2}\vDelta_k(\vv)$. Under the $\mY$-mechanism, power flow consistency implies $\vxi_k(\vv) = \tilde\vxi_k$ (the drawn log-load at bus~$k$), so $\vxi_k(\vv) \sim \gN(\vmu^{(\ell)},\mSigma^{(\ell)})$ and consequently $\vz_k^{\mY}\sim\gN(\vzero, \mI_T)$ exactly, independently across buses. Substituting $\vDelta_k = (\mSigma^{(\ell)})^{1/2}\vpsi_k$ and $\bar{\vxi}_k - \vmu^{(\ell)} = (\mSigma^{(\ell)})^{1/2}\bar{\vz}_k$ with $\bar{\vz}_k = \vz_k^{\mY} - \tfrac{1}{2}\vpsi_k$ into~\eqref{eq:midpoint_formula_app}:
\begin{align}
\Lambda_I^{(k)} &= -\vpsi_k(\vv)^\top(\bar{\vz}_k + \vbeta^{(\ell)}) \nonumber\\
&= -\vpsi_k(\vv)^\top\vz_k^{\mY} + \tfrac{1}{2}\|\vpsi_k(\vv)\|^2 - \vpsi_k(\vv)^\top\vbeta^{(\ell)}.
\label{eq:per_bus_whitened}
\end{align}

Note that $\vpsi_k(\vv)$ is a function of the drawn sample through $\vv$: the counterfactual implied log-load $\vxi'_k(\vv) = \log\vp_k^{\mY'}(\vv)$ is a deterministic function of the released voltage, and $\vv$ itself is determined by the drawn loads via the power flow map. Consequently $\vpsi_k(\vv)$ is \emph{not} independent of $\vz_k^{\mY}$, which precludes analyzing $\vpsi_k^\top\vz_k^{\mY}$ as a scalar Gaussian. We instead bound $\|\vpsi_k(\vv)\|$ uniformly in $\vv$ and apply Cauchy--Schwarz.

\emph{Step 3: Uniform bound on $\|\vpsi_k(\vv)\|^2$.}
By Lemma~\ref{lem:log_lipschitz}, $|[\vDelta_k(\vv)]_t| \le |p_k^{\mY}(\vv_t) - p_k^{\mY'}(\vv_t)|/p_{\min}^{(\ell)}$. Since the generation terms cancel, $|p_k^{\mY}(\vv_t) - p_k^{\mY'}(\vv_t)| = |\Re\{[\vv_t]_k[\overline{\Delta\mY}\bar{\vv}_t]_k\}| \le V_{\max}^2\sqrt{d_{\max}}\|(\Delta\mY)_{k,:}\|_2$, so $\|\vDelta_k(\vv)\|_\infty \le d_\ell\|(\Delta\mY)_{k,:}\|_2$ uniformly in $\vv\in\gS_0^T$.
Applying the precision-sum bound:
\begin{align}
\|\vpsi_k(\vv)\|^2 &= \vDelta_k(\vv)^\top(\mSigma^{(\ell)})^{-1}\vDelta_k(\vv) \nonumber\\
&= \sum_{t,t'}[(\mSigma^{(\ell)})^{-1}]_{tt'}[\vDelta_k]_t[\vDelta_k]_{t'} \nonumber\\
&\le \|\vDelta_k(\vv)\|_\infty^2\,\gamma^{(\ell)} \le d_\ell^{\,2}\|(\Delta\mY)_{k,:}\|_2^2\,\gamma^{(\ell)}.
\label{eq:psi_bound}
\end{align}
Aggregating over buses with $\sum_{k\in\gC_\ell}\|(\Delta\mY)_{k,:}\|_2^2 \le (\kappa_{\mathrm{Kron}} r)^2$ (by Corollary~\ref{cor:kron_perturbation}) for all $\vv\in\gS_0^T$ it holds:
\begin{equation}
\label{eq:psi_agg_bound}
\sum_{k=1}^n \|\vpsi_k(\vv)\|^2 \le (\kappa_{\mathrm{Kron}} r)^2\sum_{\ell=1}^L d_\ell^{\,2}\,\gamma^{(\ell)} = \bar\psi^{\,2}.
\end{equation}

\emph{Step 4: Cauchy--Schwarz and $\chi^2$ tail.}
Summing~\eqref{eq:per_bus_whitened} $\forall k$:
\begin{equation}
\label{eq:lambdaI_split}
\Lambda_I = \underbrace{-\sum_k\vpsi_k(\vv)^\top\vz_k^{\mY}}_{=:\,Z(\vv)} + \underbrace{\tfrac{1}{2}\sum_k\|\vpsi_k(\vv)\|^2 - \sum_k\vpsi_k(\vv)^\top\vbeta^{(\ell_k)}}_{=:\,B_I(\vv)}.
\end{equation}

Let $\vz^{\mY} := (\vz_1^{\mY},\ldots,\vz_n^{\mY}) \in \R^{nT}$, so $\|\vz^{\mY}\|_2^2 \sim \chi^2_{nT}$ exactly. Cauchy--Schwarz on the stacked vectors gives
\begin{align}
|Z(\vv)| \le \sqrt{\textstyle\sum_k\|\vpsi_k(\vv)\|^2}\cdot\|\vz^{\mY}\|_2 \le \bar\psi\cdot\|\vz^{\mY}\|_2,
\label{eq:Z_CS}
\end{align}
using~\eqref{eq:psi_agg_bound}. The Laurent--Massart inequality~\cite{laurent2000adaptive} gives, for any $x>0$, $\Pr[\|\vz^{\mY}\|_2^2 \ge nT + 2\sqrt{nT\,x} + 2x] \le e^{-x}$. Setting $x=\log(1/\delta_R)$ yields $\|\vz^{\mY}\|_2 \le \tau(\delta_R)$ with probability at least $1-\delta_R$. On this event, $|Z(\vv)| \le \bar\psi\,\tau(\delta_R)$.

For the deterministic bias $B_I(\vv)$, the uniform bound~\eqref{eq:psi_agg_bound} gives $\tfrac{1}{2}\sum_k\|\vpsi_k(\vv)\|^2 \le \tfrac{1}{2}\bar\psi^{\,2}$. For the $\vbeta$-term, Cauchy--Schwarz over buses within each class ($\sum_{k\in\gC_\ell}\|(\Delta\mY)_{k,:}\|_2 \le \sqrt{|\gC_\ell|}\,\kappa_{\mathrm{Kron}} r$) combined with $\|\vpsi_k(\vv)\|\le d_\ell\|(\Delta\mY)_{k,:}\|_2\sqrt{\gamma^{(\ell)}}$ and $\|\vbeta^{(\ell)}\| = \sqrt{\vone^\top\mSigma^{(\ell)}\vone}$ gives
\begin{equation}
\Big|\sum_k\vpsi_k(\vv)^\top\vbeta^{(\ell_k)}\Big| \le \kappa_{\mathrm{Kron}} r\sum_\ell d_\ell\sqrt{\gamma^{(\ell)}|\gC_\ell|}\sqrt{\vone^\top\mSigma^{(\ell)}\vone}.
\label{eq:beta_bound_app}
\end{equation}

Combining yields~\eqref{eq:termI_bound}. \qed

\subsection{Proof of Proposition~\ref{prop:termII}}
\label{app:termII_proof}

\begin{proof}
\emph{Wirtinger factorization.} The factorization $\mJ = \mD(\vv)\tilde{\mM}$ with $\mD(\vv)$ as in~\eqref{eq:Dv_def} yields~\eqref{eq:Mtilde_def}. From~\eqref{eq:DG_ift}, at any $\vv\in\gS_0\cap\gM_\mY$, $DG_\mY(G_\mY^{-1}(\vv)) = -(\mJ^{\mathrm{eff}}_{\tilde F_\mY}(\vv))^{-1}\mR$ and similarly for $\mY'$. With $\Delta\mJ^{\mathrm{eff}} := \mJ^{\mathrm{eff}}_{\tilde F_{\mY'}} - \mJ^{\mathrm{eff}}_{\tilde F_\mY}$, the Neumann identity gives $(\mJ^{\mathrm{eff}}_{\tilde F_{\mY'}})^{-1} = (\mI + \mK)^{-1}(\mJ^{\mathrm{eff}}_{\tilde F_\mY})^{-1}$ with $\mK := (\mJ^{\mathrm{eff}}_{\tilde F_\mY})^{-1}\Delta\mJ^{\mathrm{eff}}$. Since $\mD(\vv)$ is independent of $\mY$, $\Delta\mJ^{\mathrm{eff}} = \mD(\vv)\Delta\tilde{\mM}$ (the volt-var contribution cancels), so
\begin{align}
\label{eq:K_def}
\mK &= \tilde{\mM}^{-1}\Delta\tilde{\mM}\in\C^{2n\times 2n}, \nonumber\\
DG_{\mY'}(G_{\mY'}^{-1}(\vv)) &= (\mI+\mK)^{-1} DG_\mY(G_\mY^{-1}(\vv)).
\end{align}

\emph{Reduction to an $n\times n$ compression.} From~\eqref{eq:K_def}, $DG_{\mY'}^\top DG_{\mY'} = DG_\mY^\top\mS^{-1}DG_\mY$ with $\mS := (\mI+\mK)(\mI+\mK)^\top$. Let $DG_\mY = \mU_G\mSigma_G\mV_G^\top$ be a thin SVD with $\mU_G\in\R^{2n\times n}$. Then $\det(DG_\mY^\top DG_\mY) = \det(\mSigma_G)^2$ and $\det(DG_{\mY'}^\top DG_{\mY'}) = \det(\mSigma_G)^2\det(\mU_G^\top\mS^{-1}\mU_G)$, so the $\det(\mSigma_G)^2$ factor cancels in the ratio:
\begin{align}
\label{eq:det_ratio_n}
\frac{\det(DG_{\mY'}^\top DG_{\mY'})}{\det(DG_\mY^\top DG_\mY)} &= \det(\mI_n - \mW), \nonumber\\
\mW &:= \mU_G^\top(\mI - \mS^{-1})\mU_G\in\C^{n\times n}.
\end{align}
As $|J_\mY(\vv)| =\sqrt{ \det(DG_\mY^\top DG_\mY)}$, Term~II per time step is
\begin{equation}
\label{eq:termII_per_step}
\Lambda_{II}^{(t)}(\vv) = \log\frac{|J_\mY(\vv)|}{|J_{\mY'}(\vv)|} = -\tfrac{1}{2}\log\det(\mI_n - \mW(\vv)).
\end{equation}

\emph{Matrix-log series on an $n\times n$ matrix.} For $k\ge 2$, Schur's inequality and the $\ell_k$--$\ell_2$ inequality give $|\mathrm{tr}(\mW^k)|\le\|\mW\|_F^k$; for $k=1$, $|\mathrm{tr}(\mW)|\le\sqrt{n}\,\|\mW\|_F$ by $\ell_1$--$\ell_2$. Let $\|\mW\|_F < 1$,
\begin{align}
|\log\det(\mI_n - \mW)| &\le \sum_{k=1}^\infty\tfrac{1}{k}|\mathrm{tr}(\mW^k)| \le \sqrt{n}\,\|\mW\|_F \nonumber\\
&+ \sum_{k=2}^\infty\tfrac{1}{k}\|\mW\|_F^k \le \frac{\sqrt{n}\,\|\mW\|_F}{1-\|\mW\|_F}. \label{eq:logdet_bound_n}
\end{align}

\emph{Bounding $\|\mW\|_F$.} Extending $\mU_G$ to a full orthogonal $\tilde\mU_G\in\C^{2n\times 2n}$, $\mU_G^\top\mA\,\mU_G$ is a principal $n\times n$ submatrix, so $\|\mW\|_F\le\|\mI - \mS^{-1}\|_F$. Writing $\mI - \mS^{-1} = \mS^{-1}(\mS - \mI)$ with $\mS - \mI = \mK + \mK^\top + \mK\mK^\top$ and using $\|\mK\|_{\mathrm{op}}\le\|\mK\|_F\le\alpha$, we have $\|\mS - \mI\|_F \le 2\alpha + \alpha^2$; combined with $\|\mS^{-1}\|_{\mathrm{op}}\le(1-\alpha)^{-2}$ (since $\sigma_{\min}(\mS)\ge(1-\alpha)^2$),
\begin{equation}
\label{eq:W_bound}
\|\mW\|_F \le \|\mI - \mS^{-1}\|_F \le \frac{\alpha(2+\alpha)}{(1-\alpha)^2}.
\end{equation}
A direct calculation shows $\alpha(2+\alpha)/(1-\alpha)^2 < 1\Leftrightarrow 4\alpha < 1$, so $\alpha < 1/4$ guarantees $\|\mW\|_F < 1$ and permits~\eqref{eq:logdet_bound_n}. The bound $\|\mK\|_F\le\alpha$ follows since $\|\Delta\tilde{\mM}\|_F\le C(\vv)\|\Delta\mY\|_F$ with $C(\vv) = \sqrt{2}(1+\|\vv\|/V_{\min})$, giving $\|\mK\|_F\le\|\tilde{\mM}^{-1}\|_\star C_\star\,\kappa_{\mathrm{Kron}} r = \alpha$ after taking the worst case over $\vv\in\gS_0$ and applying Corollary~\ref{cor:kron_perturbation}.

\emph{Combining.} Plugging~\eqref{eq:W_bound} into~\eqref{eq:logdet_bound_n} and using $1 - \|\mW\|_F\ge(1-4\alpha)/(1-\alpha)^2$ gives $|\log\det(\mI_n - \mW)| \le \sqrt{n}\,\alpha(2+\alpha)/(1-4\alpha)$. By~\eqref{eq:termII_per_step}, $|\Lambda_{II}^{(t)}|\le\tfrac{1}{2}\sqrt{n}\,\alpha(2+\alpha)/(1-4\alpha)$ per time step; summing over $T$ steps yields~\eqref{eq:termII_trajectory}.
\end{proof}

\subsection{Proof of Theorem~\ref{thm:approx_dp}}
\label{app:thm_approx_proof}
\begin{proof}
By the LLR decomposition~\eqref{eq:llr_decomp}, $\Lambda = \Lambda_I + \Lambda_{II}$. Under $\alpha < 1/4$, Proposition~\ref{prop:termII} gives $|\Lambda_{II}| \le \bar{\Lambda}_{\mathrm{II}} = T\sqrt{n}\,\alpha(2+\alpha)/(2(1-4\alpha))$ deterministically. By Proposition~\ref{prop:termI}, with probability at least $1-\delta$ over the DP-GMM,
\[
|\Lambda_I| \le \bar\psi\,\tau(\delta) + \tfrac{1}{2}\bar\psi^{\,2} + \kappa_{\mathrm{Kron}}\,r\sum_{\ell=1}^L d_\ell\sqrt{\gamma^{(\ell)}|\gC_\ell|}\sqrt{\vone^\top\mSigma^{(\ell)}\vone}.
\]
Combining, on the same event of probability $\ge 1-\delta$,
\[
|\Lambda| \le \bar\psi\,\tau(\delta) + B,
\]
where $B := T\sqrt{n}\,\alpha(2+\alpha)/(2(1-4\alpha)) + \tfrac{1}{2}\bar\psi^{\,2} + \kappa_{\mathrm{Kron}}\,r\sum_{\ell=1}^L d_\ell\sqrt{\gamma^{(\ell)}|\gC_\ell|}\sqrt{\vone^\top\mSigma^{(\ell)}\vone}$ is entirely deterministic. Setting $\varepsilon = B + \bar\psi\,\tau(\delta)$ establishes $\Pr[|\Lambda|>\varepsilon]\le\delta$, i.e.\ $(\varepsilon,\delta)$-PDP by Definition~\ref{def:pdp}.
\end{proof}

\subsection{Closed-Form Bound on $\|\tilde{\mM}^{-1}\|_\star$}
\label{app:Mtilde_closed_form_proof}

Recall from Section~\ref{sec:termII} that $\|\tilde{\mM}^{-1}\|_\star := \sup_{\vv\in\gS_0,\,\mY'\in\mathcal{N}_r(\mY_{\mathrm{full}})} \|\tilde{\mM}(\vv,\mY')^{-1}\|_{\mathrm{op}}$. The following remark gives an analytical upper bound on this worst-case quantity in terms of network parameters known to the utility.

\begin{remark}
\label{rem:Mtilde_closed_form}
Define the voltage deviation bound $\Delta_\infty := \max(V_{\max}-1,\,1-V_{\min})$ and the quantities
\begin{align}
\label{eq:sigma_M_flat}
\sigma_{\tilde{\mM}} &:= \sigma_{\min}(\mY) - \|\mY\vone+\vb\|_\infty, \\
\label{eq:C3_def}
C_3 &:= \frac{\|\mY\|_{\infty\to\infty} + \|\mY\vone+\vb\|_\infty}{V_{\min}}, \\
\label{eq:Cvv_def}
C_{\mathrm{vv}} &:= \frac{h_{\max}^g}{V_{\min}}\max_{k\in\gG}\gamma_k\,\|\phi_k'\|_\infty,
\end{align}
where $\|\mY\|_{\infty\to\infty} := \max_k\sum_j|\mY_{kj}|$ is the row-sum norm, $h_{\max}^g := \sup_{\tau\in[T]} h_\tau^g$ is the peak irradiance, and $\|\phi_k'\|_\infty := \sup_{u\in[V_{\min},V_{\max}]}|\phi_k'(u)|$ is the maximum volt-var slope at bus $k$. Without volt-var coupling,
\begin{align}
\label{eq:Mtilde_inv_bound}
\|\tilde{\mM}^{-1}\|_\star \le \frac{1}{\sigma_{\tilde{\mM}} - C_3\,\Delta_\infty - C_\star\,\kappa_{\mathrm{Kron}}\,r},
\end{align}
and with volt-var coupling, the same bound holds with $\sigma_{\tilde{\mM}}$ replaced by $\sigma_{\tilde{\mM}} - C_{\mathrm{vv}}$, i.e.,
\begin{align}
\label{eq:Mtilde_eff_inv_bound}
\|\tilde{\mM}_{\mathrm{eff}}^{-1}\|_\star \le \frac{1}{\sigma_{\tilde{\mM}} - C_3\,\Delta_\infty - C_{\mathrm{vv}} - C_\star\,\kappa_{\mathrm{Kron}}\,r},
\end{align}
provided the respective denominators are positive.
\end{remark}

\begin{proof}At $\vv=\vone$, the normalized Jacobian decomposes as $\tilde{\mM}(\vone,\mY) = \tilde{\mM}_0 + \tilde{\mM}_1$, where
\begin{align*}
\tilde{\mM}_0 &:= \begin{pmatrix} \vzero & \bar{\mY} \\ \mY & \vzero \end{pmatrix}, \\
\tilde{\mM}_1 &:= \begin{pmatrix} \diag(\bar{\mY}\vone+\bar{\vb}) & \vzero \\ \vzero & \diag(\mY\vone+\vb) \end{pmatrix}.
\end{align*}
The block anti-diagonal matrix $\tilde{\mM}_0$ has singular values equal to those of $\mY$ (each with multiplicity two), so $\sigma_{\min}(\tilde{\mM}_0) = \sigma_{\min}(\mY)$. The block diagonal matrix $\tilde{\mM}_1$ satisfies $\|\tilde{\mM}_1\|_{\mathrm{op}} = \|\mY\vone+\vb\|_\infty$. By Weyl's inequality,
\begin{align}
\label{eq:sigma_Mtilde_flat}
\sigma_{\min}\!\bigl(\tilde{\mM}(\vone,\mY)\bigr) &\ge \sigma_{\min}(\mY) - \|\mY\vone+\vb\|_\infty \nonumber\\
&= \sigma_{\tilde{\mM}}.
\end{align}
At a general $\vv\in\gS_0$, only the diagonal blocks of $\tilde{\mM}$ change; entry $i$ of the $(1,1)$-block changes from $(\bar{\mY}\vone+\bar{\vb})_i$ to $(\bar{\mY}\bar{\vv}+\bar{\vb})_i/v_i$, and the difference satisfies
\begin{equation}
\biggl|\frac{(\bar{\mY}\bar{\vv}+\bar{\vb})_i}{v_i} - (\bar{\mY}\vone+\bar{\vb})_i\biggr| \le C_3\,\|\vv-\vone\|_\infty,
\end{equation}
with $C_3$ as in~\eqref{eq:C3_def}. This bounds the maximum diagonal entry, so $\|\tilde{\mM}(\vv,\mY)-\tilde{\mM}(\vone,\mY)\|_{\mathrm{op}} \le C_3\,\Delta_\infty$. For $\mY_{\mathrm{full}}'\in\mathcal{N}_r(\mY_{\mathrm{full}})$, Corollary~\ref{cor:kron_perturbation} gives $\|\mY-\mY'\|_F < \kappa_{\mathrm{Kron}} r$, so $\|\tilde{\mM}(\vv,\mY)-\tilde{\mM}(\vv,\mY')\|_{\mathrm{op}} \le C_\star\,\kappa_{\mathrm{Kron}}\,r$. Three applications of Weyl's perturbation theorem give
\begin{equation}
\sigma_{\min}\!\bigl(\tilde{\mM}(\vv,\mY')\bigr) \ge \sigma_{\tilde{\mM}} - C_3\,\Delta_\infty - C_\star\,\kappa_{\mathrm{Kron}}\,r,
\end{equation}
and since $\|\tilde{\mM}^{-1}\|_\star = 1/\inf_{\vv,\mY'}\sigma_{\min}(\tilde{\mM}(\vv,\mY'))$, the bound~\eqref{eq:Mtilde_inv_bound} follows.

\emph{Volt-var coupling bound.} With volt-var coupling, $\tilde{\mM}_{\mathrm{eff}} = \tilde{\mM} - \mD(\vv)^{-1}\mJ_{\Gamma\vh}$. Since $\mJ_{\Gamma\vh}$ is block-diagonal (each PV bus couples only to itself through the volt-var curve), $\mD(\vv)^{-1}\mJ_{\Gamma\vh}$ is also block-diagonal with entries bounded, for generation bus $k\in\gG$, by
\begin{equation}
\frac{\gamma_k h_t^g |\phi_k'(|v_k|)|}{|v_k|} \le \frac{\gamma_k h_{\max}^g \|\phi_k'\|_\infty}{V_{\min}}.
\end{equation}
Therefore $\|\mD(\vv)^{-1}\mJ_{\Gamma\vh}\|_{\mathrm{op}} \le C_{\mathrm{vv}}$ with $C_{\mathrm{vv}}$ as in~\eqref{eq:Cvv_def}. By one further application of Weyl's inequality, $\sigma_{\min}(\tilde{\mM}_{\mathrm{eff}}) \ge \sigma_{\min}(\tilde{\mM}) - C_{\mathrm{vv}}$.
\end{proof}

\subsection{Sensitivity Derivations for Baseline Mechanisms}
\label{app:sensitivity_proofs}

We derive the per-sample $\ell_2$ sensitivities reported in Table~\ref{tab:sensitivities}. For both derivations we use the following fact: the Wirtinger factorization $\mJ = \mD(\vv)\tilde{\mM}$ (Section~\ref{sec:termII}) gives
\begin{equation}
\label{eq:Jinv_Mtilde}
\|\mJ_{\tilde{F}_\mY}(\vv)^{-1}\|_{\mathrm{op}} \le \frac{\|\tilde{\mM}^{-1}\|_\star}{V_{\min}}.
\end{equation}

\begin{proof}[Derivation of $\Delta_2^{(\mY)}$]
Fix a voltage trajectory $\vv_t\in\gS_0$ and consider two $r$-adjacent full admittance matrices $\mY_{\mathrm{full}},\mY_{\mathrm{full}}'$ with $\|\mY_{\mathrm{full}}-\mY_{\mathrm{full}}'\|_F < r$. By Corollary~\ref{cor:kron_perturbation}, the induced reduced perturbation satisfies $\|\Delta\mY\|_F < \kappa_{\mathrm{Kron}} r$. By the implicit function theorem,
\begin{equation}
\vv_t(\mY) - \vv_t(\mY') \approx \mJ_{\tilde{F}_\mY}(\vv_t)^{-1}\bigl(\tilde{F}_{\mY}(\vv_t) - \tilde{F}_{\mY'}(\vv_t)\bigr),
\end{equation}
and the injection perturbation satisfies $\|\tilde{F}_{\mY}(\vv_t) - \tilde{F}_{\mY'}(\vv_t)\|_2 \le V_{\max}^2\sqrt{n}\,\|\Delta\mY\|_F \le V_{\max}^2\sqrt{n}\,\kappa_{\mathrm{Kron}}\,r$. Applying the Jacobian bound $\|\mJ^{-1}\|_{\mathrm{op}} \le \|\tilde{\mM}^{-1}\|_\star / V_{\min}$ from~\eqref{eq:Jinv_Mtilde} yields~\eqref{eq:sens_voltage_Y}.
\end{proof}

\begin{proof}[Derivation of $\Delta_2^{(\mathrm{load}\to v)}$]
A single user's load change perturbs the real-representation injection vector by at most $\sqrt{2}\,\Delta_{\mathrm{load}}$ in $\ell_2$ norm (the factor $\sqrt{2}$ accounts for real and imaginary parts). By the implicit function theorem, $\|\delta\vv\|_2 \le \|\mJ^{-1}\|_{\mathrm{op}}\,\sqrt{2}\,\Delta_{\mathrm{load}}$. Applying~\eqref{eq:Jinv_Mtilde} gives~\eqref{eq:sens_voltage_load}.
\end{proof}

\end{appendix}

\end{document}

%% file: math_commands.tex
\usepackage{xspace}

\newcommand{\norm}[1]{\left\lVert#1\right\rVert}
\newcommand{\abs}[1]{\left|#1\right|}

\DeclareMathOperator{\diag}{diag}

\newcommand{\C}{\mathbb{C}}

\newcommand{\Real}{\mathrm{Re}}
\newcommand{\Imag}{\mathrm{Im}}

\usepackage{tikz}
\usepackage[normalem]{ulem}

\usepackage{amsmath,amsfonts,bm}
\usepackage{amssymb}
\usepackage{centernot}
\usepackage{enumitem}

\usepackage{amsthm}
\theoremstyle{plain}
\newtheorem{remark}{Remark}
\newtheorem{lemma}{Lemma}
\newtheorem{assumption}{Assumption}

\newtheorem{proposition}{Proposition}
\newtheorem{corollary}{Corollary}

\newtheorem{theorem}{Theorem}
\newtheorem{definition}{Definition}

\def\1{\bm{1}}

\def\vDelta{{\bm{\Delta}}} 

\def\vbeta{{\bm{\beta}}}
\def\vxi{{\bm{\xi}}}
\def\vzero{{\bm{0}}}
\def\vone{{\bm{1}}}
\def\vmu{{\bm{\mu}}}

\def\vb{{\bm{b}}}

\def\vh{{\bm{h}}}

\def\vp{{\bm{p}}}
\def\vq{{\bm{q}}}

\def\vs{{\bm{s}}}

\def\vv{{\bm{v}}}

\def\vz{{\bm{z}}}

\def\mA{{\bm{A}}}

\def\mD{{\bm{D}}}

\def\mI{{\bm{I}}}
\def\mJ{{\bm{J}}}
\def\mK{{\bm{K}}}

\def\mM{{\bm{M}}}

\def\mR{{\bm{R}}}
\def\mS{{\bm{S}}}

\def\mU{{\bm{U}}}
\def\mV{{\bm{V}}}
\def\mW{{\bm{W}}}

\def\mY{{\bm{Y}}}

\def\mPhi{{\bm{\Phi}}}

\def\mSigma{{\bm{\Sigma}}}
\def\mTheta{{\bm{\Theta}}}

\DeclareMathAlphabet{\mathsfit}{\encodingdefault}{\sfdefault}{m}{sl}
\SetMathAlphabet{\mathsfit}{bold}{\encodingdefault}{\sfdefault}{bx}{n}

\def\gC{{\mathcal{C}}}

\def\gG{{\mathcal{G}}}

\def\gM{{\mathcal{M}}}
\def\gN{{\mathcal{N}}}

\def\gR{{\mathcal{R}}}
\def\gS{{\mathcal{S}}}

\def\gY{{\mathcal{Y}}}
\def\gZ{{\mathcal{Z}}}

\newcommand{\R}{\mathbb{R}}

